\documentclass[11pt]{article}
\usepackage[utf8]{inputenc}
\usepackage{amsmath,amsfonts,amssymb}
\usepackage{amsthm}
\usepackage{bbm}
\usepackage{bm}
\usepackage{latexsym}
\usepackage[noadjust]{cite}
\usepackage{graphicx}
\usepackage{hyperref}
\usepackage{xcolor}
\usepackage{dirtytalk}
\usepackage{mathtools}
\usepackage[margin=1in]{geometry}
\usepackage{thm-restate}
\usepackage{enumitem}
\usepackage{tikz-cd}

\usepackage[linesnumbered,ruled]{algorithm2e}

\newtheorem{theorem}{Theorem}[section]
\newtheorem{corollary}[theorem]{Corollary}
\newtheorem{lemma}[theorem]{Lemma}

\newtheorem{proposition}[theorem]{Proposition}
\newtheorem{claim}[theorem]{Claim}
\theoremstyle{definition}
\newtheorem{definition}[theorem]{Definition}

\usepackage[nameinlink]{cleveref}
\Crefname{claim}{Claim}{Claims}
\crefname{theorem}{Theorem}{Theorems}
\crefname{claim}{Claim}{Claims}
\crefname{proposition}{Proposition}{Propositions}
\crefname{definition}{Definition}{Definition}
\crefname{lemma}{Lemma}{Lemma}
\crefname{corollary}{Corollary}{Corollaries}
\crefname{ineq}{inequality}{inequalities}
\Crefname{equation}{Equation}{Equations}

\newcommand{\N}{\mathbb{N}}

\newcommand{\Z}{\mathbb{Z}}

\newcommand{\ra}{\rightarrow}
\newcommand{\la}{\leftarrow}

\DeclareMathOperator*{\E}{\mathbb{E}}

\newcommand{\poly}{\operatorname{poly}}
\newcommand{\polylog}{\operatorname{polylog}}

\renewcommand{\L}{\bm{\mathsf{L}}}
\newcommand{\RL}{\bm{\mathsf{RL}}}
\newcommand{\NL}{\bm{\mathsf{NL}}}
\newcommand{\BPL}{\bm{\mathsf{BPL}}}
\renewcommand{\P}{\bm{\mathsf{P}}}
\newcommand{\BPP}{\bm{\mathsf{BPP}}}
\newcommand{\ZPsL}{\bm{\mathsf{ZP^*L}}}
\newcommand{\SPACE}{\bm{\mathsf{SPACE}}}

\newcommand{\prBPL}{\bm{\mathsf{prBPL}}}
\newcommand{\prRL}{\bm{\mathsf{prRL}}}
\newcommand{\promBPL}{\bm{\mathsf{promise}\textbf{-}\mathsf{BPL}}}

\newcommand{\va}{v_{\text{acc}}}

\newcommand{\vs}{v_{st}}

\newcommand{\eps}{\varepsilon}

\newcommand{\zo}{\{0,1\}}

\newcommand{\cF}{\mathcal{F}}
\newcommand{\SAMP}{\mathtt{SAMP}}
\newcommand{\BIAS}{\mathtt{BIAS}}
\newcommand{\DEC}{\mathtt{DEC}}
\newcommand{\MAJ}{\mathtt{MAJ}}

\newcommand{\SUC}{\mathtt{SUC}}
\newcommand{\ADV}{\mathtt{ADV}}

\newcommand{\tp}{\tilde{p}}

\DeclareMathOperator{\Cov}{Cov}
\DeclareMathOperator{\Var}{Var}

\newcommand{\circC}{\mathcal{C}}
\newcommand{\circB}{\mathcal{B}}
\newcommand{\circF}{\mathcal{F}}

\newif\ifdraft
\draftfalse

\newif\ifnames
\namestrue

\title{Certified Hardness vs. Randomness for Log-Space}
\ifnames{
\author{Edward Pyne\thanks{Supported by an Akamai Presidential Fellowship. Part of this work was done while visiting the Simons Institute Program on Meta-Complexity.}\\MIT\\\texttt{epyne@mit.edu} \and Ran Raz\thanks{Supported by a Simons Investigator Award and by the National Science Foundation grant No. CCF-2007462.}\\ Princeton University \\ \texttt{ranr@cs.princeton.edu} \and Wei Zhan\thanks{Supported by a Simons Investigator Award and by the National Science Foundation grant No. CCF-2007462.}\\ Princeton University \\ \texttt{weizhan@cs.princeton.edu}}\fi

\begin{document}
\begin{titlepage}
\maketitle
\begin{abstract}
Let $\mathcal{L}$ be a language that can be decided in linear space and let $\epsilon >0$ be any constant. Let $\cal{A}$ be the exponential hardness assumption that for every $n$, membership in $\cal{L}$ for inputs of length~$n$ cannot be decided by circuits of size smaller than $2^{\epsilon n}$. 
We prove that for every function $f :\{0,1\}^* \rightarrow \{0,1\}$, computable by a randomized logspace algorithm $R$, there exists a deterministic logspace algorithm $D$ (attempting to compute $f$), such that on every input $x$ of length $n$, the algorithm $D$ outputs one of the following:
\begin{enumerate}
\item 
The correct value $f(x)$.
\item
The string: ``I am unable to compute $f(x)$ because the hardness assumption $\cal{A}$ is false'', followed by a (provenly correct) circuit of size smaller than $2^{\epsilon n'}$ for membership in $\cal{L}$ for inputs of length~$n'$, for some $n' = \Theta (\log n)$; that is, a circuit that refutes $\cal{A}$.
\end{enumerate}
Moreover, $D$ is explicitly constructed, given $R$.

We note that previous works on the hardness-versus-randomness paradigm give derandomized algorithms that rely blindly on the hardness assumption. If the hardness assumption is false, the algorithms may output incorrect values, and thus a user cannot trust that an output given by the algorithm is correct. Instead, our algorithm $D$ verifies the computation so that it never outputs an incorrect value. Thus, if $D$ outputs a value for $f(x)$, that value is certified to be correct. 
Moreover, if $D$ does not output a value for $f(x)$, it alerts that the hardness assumption was found to be false, and refutes the assumption.

Our next result is a universal derandomizer for $\BPL$ (the class of problems solvable by bounded-error randomized logspace algorithms)\footnote{Our result is stated and proved for $\promBPL$, but we ignore this difference in the abstract.}: 
We give a deterministic algorithm $U$ that takes as an input a randomized logspace algorithm $R$ and an input $x$ and simulates the computation of $R$ on $x$, deteriministically. Under the widely believed assumption $\BPL=\L$, the space used by $U$ is at most $C_R \cdot \log n$ (where $C_R$ is a constant depending on~$R$). Moreover, for every constant $c \geq 1$, if $\BPL\subseteq \SPACE[(\log(n))^{c}]$ then the space used by $U$ is at most $C_R  \cdot (\log(n))^{c}$.

Finally, we prove that if optimal hitting sets for ordered branching programs exist then there is a deterministic logspace algorithm that, given a black-box access to an ordered branching program $B$ of size $n$, estimates the probability that $B$ accepts on a uniformly random input.
This extends the result of (Cheng and Hoza CCC 2020), who proved that an optimal hitting set implies a white-box two-sided derandomization. 
\end{abstract}

\vfill
\thispagestyle{empty}
\end{titlepage}
\newpage

\section{Introduction}

In a recent work, Girish, Raz and Zhan studied the power of untrusted randomness~\cite{GRZ23}. One of their main observations was that randomized logspace computations are verifiable using only $O(\log n)$ random bits.
More precisely, 
every problem in $\BPL$ has a streaming proof between a
randomized logspace prover and a randomized logspace verifier, where the verifier uses only $O(\log n)$ random bits
and has a read-once one-way access to the proof that is streamed by the prover.
In other words, the prover provides a polynomial-length proof that is streamed to the verifier and the verifier
can check whether the computation was performed correctly using only $O(\log n)$ random bits.

This raises the following intriguing possibility. Try to replace the random string of the prover by, say, the digits of $\pi$. In most cases, that should work and the computation should be performed correctly, as the digits of $\pi$ seem unrelated to most computations. In the rare cases that the computation is not performed correctly, the verifier will figure that out, as the verification will fail with high probability, so no harm is done. Moreover, since the digits of $\pi$ can be generated deterministically in small space, the prover is now deterministic so the verifier can fully simulate the prover. Since the verifier uses only $O(\log n)$ random bits, the verifier can just try all possibilities for these random bits so that the verifier is also deterministic\footnote{Derandomizing the verifier by trying all possibilities for its random bits is not possible when the prover is randomized, or when the prover cannot be simulated by the verifier, since the verifier needs multi-access to the output of the prover in order to do that.}, and thus the entire interaction is now simulated by a deterministic logspace algorithm.

This approach won't derandomize all randomized logspace computations,
since the digits of $\pi$ can be generated by a small space algorithm. The digits of $\pi$ were not designed to fool randomized computations. The next logical step is to try to use sequences that were designed to fool randomized computations, namely, candidate constructions of pseudorandom generators, such as pseudorandom generators that are based on the hardness-versus-randomness paradigm~\cite{Sha,Yao,BM84,NW94,IW97,STV,KM02}. Such pseudorandom generators fool randomized computations, within a certain complexity class, assuming that certain widely-believed hardness assumptions hold.

Let $G$ be a candidate construction for a pseudorandom generator, designed to fool randomized logspace computations, and assume that $G$ uses logarithmic space and $O(\log n)$ random bits.
We can try to replace the random string of the prover by pseudorandom sequences that are generated by $G$. 
Since, if we do so, both the prover and the verifier use a logarithmic number of random bits, 
the verifier can simulate the entire interaction by a deterministic logspace algorithm. If one of the $\poly(n)$ possibilities for the $O(\log n)$ random bits of the generator results in a valid proof that the computation was performed correctly, the verifier will figure that out and accept that computation. If all $\poly(n)$ possibilities fail, the verifier will alert that the generator failed. Thus, the algorithm never outputs an incorrect value.

If the generator $G$ is based on the hardness-versus-randomness paradigm, a failure of the generator implies that the hardness assumption that the generator is based on is false. 
Moreover, proofs that are based on the hardness-versus-randomness paradigm are typically constructive, in the sense that they show that if the generator fails then one can 
construct a circuit that refutes the hardness assumption.
If we can prove that constructing that circuit can be done in deterministic logspace then the verifier can obtain a circuit that refutes the hardness assumption that $G$ is based on.

We use a variant of the hardness-versus-randomness pseudorandom generator of Klivans and van Melkebeek~\cite{KM02} that builds on~\cite{NW94,IW97,STV} to derandomize $\BPL$ (assuming an exponential hardness assumption). Based on this generator, we obtain the following result. 

\begin{theorem}\label{thm:hard}
Let $\cal{L}$ be a language that can be decided in linear space and let $\epsilon >0$ be a constant. Let $\cal{A}$ be the exponential hardness assumption that for every $n$, membership in $\cal{L}$ for inputs of length~$n$ cannot be decided by circuits of size smaller than $2^{\eps n}$. 
Let $f :\{0,1\}^* \rightarrow \{0,1\}$ be a function computable by a randomized logspace algorithm $R$.  Then, there exists a deterministic logspace algorithm $D$ (explicitly given from $R$), such that on every input $x$ of length $n$, the algorithm $D$ outputs one of the following:
\begin{enumerate}
\item 
The correct value $f(x)$.
\item
The string: ``Unable to compute $f(x)$ because the hardness assumption $\cal{A}$ is false'', followed by a (provenly correct) circuit of size smaller than $2^{\eps n'}$ for membership in $\cal{L}$ for inputs of length~$n'$, for some $n' = \Theta (\log n)$; that is, a circuit that refutes $\cal{A}$.
\end{enumerate}
\end{theorem}

In other words, while the algorithms given by all previous derandomization results based on the hardness-versus-randomness paradigm rely blindly on the hardness assumption, and may output incorrect values if the hardness assumption is false, our algorithm $D$ never outputs an incorrect value: If the hardness assumption is true, $D$ always outputs the correct value $f(x)$. If the hardness assumption is false $D$ still outputs the correct value $f(x)$, or alerts that the hardness assumption is false, and refutes the assumption. 

In particular, if the hardness assumption used in Theorem~\ref{thm:hard} is true (and there are several such assumptions that are widely believed to be true), Theorem~\ref{thm:hard} gives a deterministic logspace algorithm that always outputs the correct value of $f(x)$ and that value is certified to be correct.
In that sense, if the hardness assumption is true, the algorithm given by Theorem~\ref{thm:hard} effectively functions as a full derandomizer for the class $\BPL$.

We note that in previous works, the, so called, reconstruction step, in which a circuit that refutes the hardness assumption is constructed (when the generator fails), required the use of randomness in multiple places and was not known to be computable in logspace.
Our main technical contribution in the proof of Theorem~\ref{thm:hard} is carefully designing the pseudorandom generator and 
proving that for that generator, all parts of the reconstruction step can be done in deterministic logspace. We view this result, that the reconstruction can be done in deterministic logspace, as a separate contribution of our work.

Let us go back to the observation that
every problem in $\BPL$ has a streaming proof between a
randomized logspace prover and a randomized logspace verifier, where the verifier uses only $O(\log n)$ random bits
and has a read-once one-way access to the proof that is streamed by the prover~\cite{GRZ23}. The proof is based on a protocol where
the prover computes and streams the probability to reach each state of the branching program, underlying a randomized algorithm, and the verifier
checks that these probabilities are consistent between each two consecutive time steps.

While we can use this approach to prove Theorem~\ref{thm:hard}, we give here a slightly different and more direct proof, where the verification is done by verifying that the distribution of each bit that the pseudorandom generator outputs, conditioned on reaching each state of the underlying branching program, is close to uniform. These conditional probabilities are computed directly by checking all possible outputs of the pseudorandom generator. This is possible because the generator uses only a logarithmic number of random bits and hence the number of possibilities is polynomial in~$n$. This approach is related to the work of Nisan~\cite{Nisan93}, who used a similar approach to check if a given polynomial-size set of strings is sufficiently random to simulate a randomized computation with high accuracy, in his proof that $\BPL\subset \ZPsL$ (where $\ZPsL$ is zero-error randomized logspace, where the machine has \textit{two-way} access to the random tape).

The discussion above implies that the output of a candidate pseudorandom generator $G$ (that uses logarithmic space and $O(\log n)$ random bits) can be verified as being sufficiently random for a given randomized logspace computation. With this in mind, it is natural to try to find a pseudorandom generator that will be sufficiently good for a given randomized logspace computation, by an exhaustive search over all possible generators (using the fact that the generator is described by a constant size Turing machine). The final goal is to obtain a universal derandomizer, that will do at least as good as the best pseudorandom generator.

We explore this idea and discover that an even stronger result can be proved. We explicitly construct a universal derandomizer $U$ for $\prBPL$ ($\promBPL$, the class of promise problems solvable by bounded-error randomized logspace algorithms) that runs in the best possible deterministic space bound on $\prBPL$. 

More precisely,
we give a deterministic algorithm $U$ that takes as an input a randomized logspace algorithm $R$ and an input $x$ and simulates the computation of $R$ on $x$. 
Under the widely believed assumption $\prBPL=\L$, the space used by $U$ is at most $C_R \cdot \log n$ (where $C_R$ is a constant depending on~$R$). More generally, for every constant $c \geq 1$, if $\prBPL\subseteq \SPACE[(\log(n))^{c}]$ then the space used by $U$ is at most $C_R  \cdot (\log(n))^{c}$.
We emphasize that the point here is that $U$ is deterministic and is explicitly given, rather than an existential result.
We remark that a similar result is not known in the time bounded case, and seems hard to obtain. We also remark that
the best currently known space bound on $\BPL$ is $\prBPL\subseteq \SPACE[(\log(n))^{1.5-o(1)}]$~\cite{SZ99,Hoza21}.

\begin{theorem} \label{universal-derandomizer}
Let $U$ be the deterministic algorithm that is explicitly given in Section~\ref{section:universal}, that takes as an input a randomized logspace algorithm $R$ and an input $x$. Assume that the probability that $R$ accepts on $x$ is either  $\leq 1/4$ or $ \geq 3/4$.  Then, if the probability that $R$ accepts on $x$ is $\leq 1/4$, the output of $U$ on input $R,x$ is 0 and if the probability that $R$ accepts on $x$ is $\geq 3/4$, the output of $U$on input $R,x$ is 1.
Moreover, for every constant $c \geq 1$, if $\prBPL\subseteq \SPACE[(\log(n))^{c}]$ then the space used by $U$ is at most $C_R  \cdot (\log(n))^{c}$, (where $C_R$ is a constant depending on~$R$).
\end{theorem}
We note that one can bound the space used by $U$, in Theorem~\ref{universal-derandomizer}, also by $C  \cdot (\log(N))^{c}$, where~$C$ is a universal constant and $N$ is an upper bound on both the length and width of the branching program underlying the computation of~$R$ on $x$
(under the assumption $\prBPL\subseteq \SPACE[(\log(n))^{c}]$).
(See Theorem~\ref{thm:universal}).

Another prior work that is related to our work, as well as to~\cite{Nisan93, GRZ23}, is the work of Cheng and Hoza~\cite{CH20}. Cheng and Hoza proved that an optimal hitting set generator (the one-sided analogue of a pseudorandom generator) for logspace would imply $\BPL=\L$ (whereas the direct conclusion of such a hitting set generator would only be $\RL=\L$)~\cite{CH20}. To prove this result, they show how to use the hitting set generator to guess (approximations of) the probability to reach each state of a branching program, and 
they then check that these probabilities are consistent between each two consecutive time steps (similarly to and prior to~\cite{GRZ23}).

The proof given by Cheng and Hoza uses the explicit description of the underlying branching program. Our final result is an extension of their result to the case where the branching program is not given explicitly, but rather one only has oracle access to it, that is, access as a black box. 

\begin{theorem}[Informal: formally stated and proved in Section~\ref{sec:BBT}]\label{thm:equiv:int}
    Assume that optimal explicit hitting set generators for width $n$, length $n$ ordered branching programs exist. Then optimal deterministic samplers for width $n$, length $n$ ordered branching programs (with oracle access to the branching program) exist.
\end{theorem}

We remark that Cheng and Hoza~\cite{CH20} prove a version of this result for \textit{constant} width branching programs (in addition to their non-black-box result on length $n$, width $n$ programs that capture $\BPL$). They state a black-box equivalence in the $\BPL$ vs $\L$ regime as an open question, which we resolve. Our result complements equivalent results in the $\BPP$ vs $\P$ regime; several prior results~\cite{ACR96,BF99,ACRT99,GVW11,CH20} show that a hitting set for general circuits implies a deterministic sampler for general circuits. Thus, we close the gap in understanding between time-bounded and space-bounded derandomization with regards to this question.

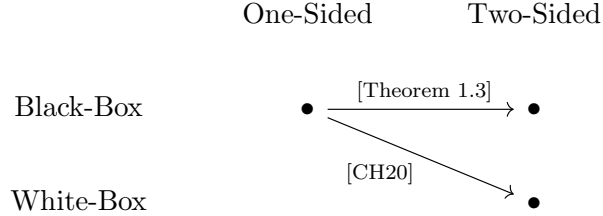
\[\begin{tikzcd}
    & {\text{One-Sided}} & {\text{Two-Sided}} \\
    {\text{Black-Box}} & \bullet & \bullet \\
    {\text{White-Box}} && \bullet
    \arrow["{\text{       [Theorem~\ref{thm:equiv:int}]}}", from=2-2, to=2-3]
    \arrow["{\text{[CH20]}}"', from=2-2, to=3-3]
\end{tikzcd}\]
We hope that our progress can eventually be used to get an equivalence in the \textit{white-box} regime, that is, that $\prRL=\L \implies \prBPL=\L$. Such a result was established in the time-bounded regime by~\cite{BF99}. 

A common theme in all of our results is that our proofs exploit, and further demonstrate, the intriguing idea that in some settings randomized logspace computations can be verified.

\subsection{Related Work}

There have been four decades of work attempting to derandomize randomized logspace, that is, prove $\BPL=\L$. This work has taken (at least) two major forms: constructions of pseudorandom generators (PRGs) and their generalizations~\cite{Nisan92,INW94,NZ96, GR14,FK18,MRT19,HZ20} and white-box derandomizations~\cite{SZ99,RR99,Reingold08, RTV06,AKMPSV20,Hoza21}. This has resulted in a varied landscape, with explicit constructions of PRGs that obtain highly nontrivial but (presumably) suboptimal seed lengths, white-box derandomizations, and candidate constructions. We emphasize that these candidate constructions consist of both generators whose security follows from a certain hardness assumption~\cite{KM02}, and candidates that are not known to follow from a hardness assumption (for instance, the XOR of two small-bias distributions has been proposed as a candidate by Reingold and Vadhan~\cite{LV17}). 

As mentioned above, besides~\cite{KM02}, the works most relevant to ours are~\cite{Nisan93,CH20,GRZ23}. All these works have an element of verification that a randomized computation was performed correctly (in various forms and for various purposes), an idea that is also central in our work.

\section{Preliminaries}
We first define notation related to pseudorandom generators and branching programs.
\begin{definition}
    Given a distribution $D$ over a space $[S]$, let $x\la D$ represent drawing $x\in [S]$ from $D$. We let $U_n$ denote the uniform distribution over $\zo^n$.
\end{definition}
\begin{definition}
    Given a pseudorandom generator (PRG) $G:\zo^s\ra\zo^n$
    and a function $f:\zo^n\rightarrow\mathbb{R}$, we use $\E[f(U_n)]$ and $\E[f(G(U_s)]$ to denote the expectation of $f$ under uniformly distributed inputs and pseudorandom inputs generated by $G$ respectively, that is,
    \[\E[f]=\E_{x\la U_n}[f(x)],\qquad \E_G[f]=\E_{y\la U_s}[f(G(y))].\]
    And we say that $G$ \textbf{$\eps$-fools} $f$ if $\left|\E[f] -\E_G[f]\right|\leq\eps$.
\end{definition}

\newcommand{\Bv}{B_{\ra v}}
\begin{definition}
    An ordered branching program (OBP) $B$ of length $n$ and width $w$ is a directed acyclic graph whose vertices (or states) are partitioned into $n+1$ layers $V_0,\ldots,V_n$ where $|V_i|\leq w$. For each $i<n$ and $v\in V_i$, there are two outgoing edges, labeled with $0$ and $1$ respectively, that leads into $V_{i+1}$. $V_0$ constains a single state $v_0$ which is the starting state, and each state in $V_n$ is labeled with a real number as the output of the branching program. Unless otherwise specified, we assume that the labels are either $0$ or $1$.

    For each $v\in V_i$, $\sigma \in \zo^k$ and $u\in V_{i+k}$, we say $B[v,\sigma]=u$ if $B$ transitions from state $v$ to state $u$ following the edges labeled by the bits in $\sigma$.
    We can think of $B$ as a function on $\zo^n$ such that for every $x\in\zo^n$, $B(x)$ is the label on the output state $B[v_0,x]$. For each $v\in V_i$, let $\Bv$ be an OBP of length $i$ and width $w$ such that $\Bv(x)=1$ if and only if $B[v_0,x_{1..i}]=v$.

    For each $v\in V_i$, let 
    \[p_{\ra v}=\Pr[B[\vs,U_i]=v], \quad\quad p_{v\ra}=\Pr[B[v,U_{n-i}]=\va].\]
\end{definition}

\section{Effective Hardness to Randomness}
We prove \Cref{thm:hard} in several stages. In the first stage, we show a testing procedure that, given a candidate PRG and an ordered branching program, either certifies that the PRG fools the branching program, or outputs a branching program that acts as a next-bit predictor for $G$. We then show how to go from such a next-bit predictor to a counterexample to the hardness assumption.

\subsection{Verifiable PRGs for Logspace}\label{subsec:whiteBoxLC}
We first show that there is a logspace verifier for PRGs (with logarithmic seed) against logspace OBPs, which detects when a PRG fails and outputs an example OBP that the PRG fails to fool. 
To formalize this, we recall the notion of a next-bit-predictor.
\begin{definition}\label{def:NBP}
    Given a function $G:\zo^s\ra\zo^n$, a branching program $T:\zo^i\ra\zo$ for $i<n$ is an $\eps$-\textbf{next-bit-predictor} for $G$ if $\Pr_{x\la U_s}[T(G(x)_{1..i})=G(x)_{i+1}]>1/2+\eps$.
\end{definition}
Note that the uniform distribution is $0$-next-bit-predictable, even for a computationally unbounded distinguisher.

We prove in this section the following lemma:
\begin{lemma}\label{lem:tester1}
    For every error function $\eps(n)$ computable in space $O(\log n)$, there is a deterministic algorithm that, given as input an OBP $B$ of length $n$ and width $w$, and the black-box oracle access to a PRG $G:\zo^s\ra\zo^n$, runs in space $O(s+\log(nw))$, and either
    \begin{enumerate}
        \item Confirms that $G$ $\eps\cdot n$-fools $B$; Or
        \item Outputs an OBP $T$ of length at most $n$ and width $w$ that is an $\eps/2$-next-bit predictor for $G$.
    \end{enumerate}
\end{lemma}

The main idea behind this proof has appeared before for different purposes~\cite{Nisan93,CH20,GRZ23}, and in fact (modifications of) all these results can be used to prove \Cref{lem:tester1}. However, we give a self-contained proof. 

To prove \Cref{lem:tester1}, we first define a series of potential distinguishers, with the property that each can be evaluated in logspace. Each distinguisher measures the bias of the next bit in the PRG upon reaching a particular state.

\begin{definition}
    Given an OBP $B$ of length $n$, for every $i<n$ and $v\in V_i$, let $N_v:\zo^{i+1}\ra \{-1,0,1\}$ be the function defined as:
    \[N_v(x) = \begin{cases}
    1 & \text{if }\Bv(x)=1 \text{ and } x_{i+1}=1\\
    -1 & \text{if }\Bv(x)=1 \text{ and } x_{i+1}=0\\
    0 & \text{otherwise.}
    \end{cases}\]
    Furthermore, $N_v$ is computable in logspace given $B$ and $v$.
\end{definition}

When $x$ is uniformly random, $\Bv(x)$ and $x_{i+1}$ are independent, and therefore $\E[N_v]=0$ for all $v$. Consequentially, our verifier checks that $|\E_G[N_v]|$ is small for all $v$, where we feed the first $i+1$ bits of the PRG output to $N_v$. We first show its soundness:

\begin{lemma}\label{lem:fool}
    Given an OBP $B$ of length $n$, suppose that for every $i$, $\sum_{v\in V_i}|\E_G[N_v]|\leq \eps$. Then $G$ $\eps\cdot n$-fools $B$.
\end{lemma}
\begin{proof}
    As every edge from layer $V_i$ goes into layer $V_{i+1}$, for every $i<n$ we have
    \begin{align*}
        &\ \sum_{v\in V_{i+1}}\left|\E[\Bv]-\E_G[\Bv]\right| \\
        \leq &\ \sum_{v\in V_i}\sum_{b\in\zo}
        \left|\Pr_{x\la U_n}[\Bv(x)=1\wedge x_{i+1}=b]-\Pr_{x\la G(U_s)}[\Bv(x)=1\wedge x_{i+1}=b]\right|.
    \end{align*}
    Notice that by the definition of $N_v$, we have
    \begin{align*}
        \E[N_v] & =\Pr_{x\la U_n}[\Bv(x)=1\wedge x_{i+1}=1]-\Pr_{x\la U_n}[\Bv(x)=1\wedge x_{i+1}=0] \\
        & = 2\Pr_{x\la U_n}[\Bv(x)=1\wedge x_{i+1}=1]-\E[\Bv] \\
        & = \E[\Bv] - 2\Pr_{x\la U_n}[\Bv(x)=1\wedge x_{i+1}=0],
    \end{align*}
    and the above holds similarly under pseudorandomness generated by $G$. Therefore we further have
    \begin{align*}
        \sum_{v\in V_{i+1}}\left|\E[\Bv]-\E_G[\Bv]\right|
        & \leq \sum_{v\in V_i}\left|\E[\Bv]-\E_G[\Bv]\right|+ \sum_{v\in V_i}\left|\E[N_v]-\E_G[N_v]\right| \\
        & = \sum_{v\in V_i}\left|\E[\Bv]-\E_G[\Bv]\right|+\sum_{v\in V_i}\left|\E_G[N_v]\right|.
    \end{align*}
    With the assumption that $\sum_{v\in V_i}|\E_G[N_v]|\leq \eps$ and the fact that $\E[B_{\ra v_0}]=\E_G[B_{\ra v_0}]=1$, we conclude that $\sum_{v\in V_n}\left|\E[\Bv]-\E_G[\Bv]\right|\leq \eps\cdot n$. As the output labels are binary, this means that $\left|\E[B]-\E_G[B]\right|\leq \eps\cdot n$, i.e. $G$ $\eps\cdot n$-fools $B$.
\end{proof}
\begin{proof}[Proof of \Cref{lem:tester1}]
    For every $i<n$, the algorithm iterates through every $v\in V_i$ and all the possible seeds for $G$, computes $\sum_{v\in V_i}|\E_G[N_v]|$ and checks if it is at most $\eps$. This can be done in space $O(s+\log(nw))$. If all such checks pass, we have by \Cref{lem:fool} that $G$ $\eps\cdot n$-fools $B$.

    Otherwise, we find some $i<n$ such that $\sum_{v\in V_i}|\E_G[N_v]|> \eps$. Let $T$ be an OBP of length $i$ that is the same as $B$ from layer $V_0$ to $V_i$, such that the output label on each $v\in V_i$ is $1$ if $\E_G[N_v]\geq 0$, and $0$ if $\E_G[N_v]<0$. Such an OBP is of size at most that of $B$, and can be constructed in space $O(s+\log nw)$. We have
    \begin{align*}
        \MoveEqLeft{\Pr_{x\la G(U_s)}[T(x_{1..i})=x_{i+1}]}\\ 
        &=
        \sum_{\substack{v\in V_i \\ \E_G[N_v]\geq 0}}\Pr_{x\la G(U_s)}[\Bv(x)=1\wedge x_{i+1}=1]
        + \sum_{\substack{v\in V_i \\ \E_G[N_v]< 0}}\Pr_{x\la G(U_s)}[\Bv(x)=1\wedge x_{i+1}=0] \\
        &= \sum_{v\in V_i}\frac{1}{2}\left(\E_G[\Bv]+\left|\E_G[N_v]\right|\right) \\
        &> \frac{1}{2}(1+\eps). \qedhere
    \end{align*}
\end{proof}

\subsection{Refutable Hardness Assumptions in Logspace}

\Cref{lem:tester1} shows that, given an alleged PRG for logspace, we can use it to either successfully derandomize a logspace comptation, or explicitly output a counterexample to the PRG. The results of the hardness-versus-randomness paradigm claim that PRGs exist under certain hardness assumptions. Combining these results with \Cref{lem:tester1}, we can derandomize logspace computations given any alleged hard function, or determine that the hardness assumption does not hold. However, \Cref{thm:hard} requires a stronger guarantee from the algorithm - if the hardness assumption does not hold, the algorithm needs to output a small circuit that falsifies this assumption. Obtaining this result is the primary contribution of this subsection.

We first recall the result of Klivans and van Melkebeek \cite{KM02}.
\begin{theorem}[\cite{KM02}]
    If there is a family of boolean functions $f\in\SPACE[n]$ that is not computable by circuits of size $2^{\varepsilon n}$ for some $\varepsilon>0$, then $\BPL=\L$.
\end{theorem}

Their proof is based on the worst-case hardness vs. randomness results by Imagliazzo and Wigderson \cite{IW97}, and shows how to execute every step in the construction of the Imagliazzo-Wigderson PRG can be executed in deterministic logspace. However, their proof (and all other proofs of the hardness vs. randomness paradigm) does not show that given a branching program (or circuit) that distinguishes the PRG from random (i.e. contradicts the original hardness assumption), there is an efficient deterministic logspace algorithm to produce a circuit for the supposedly hard function. This is for two reasons. First, the conversion from a distinguisher to a next bit predictor (which we address in \Cref{lem:tester1}). Even once we obtain such a predictor, prior approaches used space- and randomness-inefficient probabilistic method arguments to go from a predictor to a worst-case correct circuit for the original function. Our primary contribution in this subsection is to carefully design the PRG and develop an efficient \textit{reconstruction} procedure, given a distinguisher for the constructed PRG. 

This leads to the following theorem:
\begin{theorem}\label{thm:hard_IW}
    For every family of boolean functions $f\in\SPACE[n]$ and $\eps>0$, there is a deterministic algorithm that, given as the input an OBP $B$ of length $n$ and width $w=n$, runs in space $O(\log n)$, and either
    \begin{enumerate}
        \item Outputs $\E[B]$ with $1/4$ error; Or
        \item Outputs a circuit $\circC$ of size $2^{\eps m}$ that computes $f$ on $\zo^m$ where $m = \Theta(\log n)$.
    \end{enumerate}
\end{theorem}
\begin{proof}
    Let $G:\zo^s\ra\zo^n$ be the generator of \Cref{thm:IWrecon} with $\eps=\eps$ and $f=f$ and let $m=m_0$ be the instance size of $f$ used to construct $G$.

    We then apply \Cref{lem:tester1} on $B$ and $G$ with $\eps=1/(4n)$. Of the two possible outcomes:
    \begin{enumerate}
        \item If it is certified that $G$ $\eps\cdot n$-fools $B$. In this case the algorithm computes and outputs $\E_G[B]$ which approximates $\E[B]$ within additive error $1/4$.
        \item Otherwise we get for some $i<n$ an explicit OBP $T$ of length $i$ and width $w$, such that $\Pr_{x\la U_s}[T(G(x)_{1..i})=G(x)_{i+1}]>\frac{1}{2}(1+\eps)$. In other words, $T$ is an $\eps/2=1/8n$ next-bit predictor against $G$ of size at most $n^2$, and $T$ can be evaluated in space $O(\log n)$. Then by \Cref{thm:IWrecon}, we can construct in space $O(\log n)$ a circuit $\circC$ for $f$ on inputs of size $m=\Theta(\log n)$ of size at most $2^{\eps m}$. \qedhere
    \end{enumerate}
\end{proof}
Now \Cref{thm:hard} follows:
\begin{proof}[Proof of \Cref{thm:hard}]
    Given a randomized logspace algorithm $R$ with error probability at most $1/10$ and input $x\in \zo^n$, let $B$ be the branching program representing how $R$ uses its random bits on input $x$, which can be constructed in logspace. By assumption $R$ uses $s=O(\log n)$ bits of space and hence at most $2^{s+O(1)}=\poly(n)$ random coins, and hence $B$ has length and width $\poly(n)$. Pad $B$ to have length and width $n^c$ and apply \Cref{thm:hard_IW} with $f=\mathcal{L}$ and $\eps:=\eps$. Then we either obtain an estimate of $\E[B]$ up to $\pm 1/4$ (which suffices to decide the language by correctness of $R$) or a counterexample (in the form of a circuit of size at most $2^{\eps n'}$ to the hardness of $\mathcal{L}$ on inputs of size $n'=\Theta(\log n^c) = \Theta(\log n)$ bits, as desired.
\end{proof}
We prove \Cref{thm:IWrecon} in the following section.

\section{Efficient Reconstructive Derandomization}
We first state our main theorem of this section.
\begin{theorem}\label{thm:IWrecon}
    Given $\eps>0$, and a family of functions $f_m:\zo^m\ra\zo\in \SPACE[m]$, there is a family of explicit generators $G:\zo^s\ra\zo^n$ with $s=O(\log n)$ computable in space $O(\log n)$, and a deterministic logspace algorithm that, given $n\in\N$ and a $1/(8n)$-next-bit predictor $\circB$ for $G$ of size at most $n^2$ which is evaluable in space $O(\log n)$, outputs a circuit $\circC$ of size $2^{\eps \cdot m_0}$ for $f_{m_0}$ with $m_0=\Theta(\log n)$. 
\end{theorem}
We prove this theorem in four stages. Following the framework of \cite{IW97}, we first assume that $f$ is a (worst-case) hard function, and construct a PRG via hardness amplifications and the Nisan-Wigderson PRGs \cite{NW94}. The detailed steps are slightly different from those in \cite{IW97}, and we adapt the following strategy:
\begin{enumerate}[label=(\alph*)]
    \item From $f$, construct (by low-degree extension) a function $f'$ that is hard-on-average on a 0.99 fraction of inputs.
    \item From $f'$, construct (by derandomized XOR Lemma) a function $f''$ (with multiple bits of output) that is hard-on-average on a $2^{-\Omega(m)}$ fraction of inputs.
    \item From $f''$, construct (by Goldreich-Levin) a function $f'''$ with single-bit output that is hard-on-average on a $1/2+2^{-\Omega(m)}$ fraction of inputs.
    \item Use $f'''$ to instantiate a Nisan-Wigderson pseudorandom generator $G:\zo^s\ra\zo^n$ for $s=O(m)$.
\end{enumerate}
We make sure that $f'$, $f''$, $f'''$ and $G$ are all computable within $O(\log n)$ space. 

Furthermore, we prove that every step can be made logspace reconstructive, in the sense that given a counterexample to the conclusion (i.e. a small circuit that obtains some advantage) we can produce a counterexample to the assumption in deterministic logspace. This requires modifying the standard reconstruction algorithms for the first three steps, all of which use randomness-inefficient applications of the probabilistic method. Over the next four subsections, we state and prove the necessary components of the reconstructive PRG, and in \Cref{subsec:combine}, combine these results to conclude \Cref{thm:IWrecon}.

\subsection{Preliminaries}
First, we recall some notation related to the advantage of circuits.
\begin{definition}
    Given $f:\zo^n\ra\zo^m$ and a circuit $\circC$, let $\SUC(\circC,f) = \Pr_{x\la U_n}[\circC(x)=f(x)]$. For $m=1$, let $\ADV(\circC,f)=2\SUC(\circC,f)-1$. Let $\ADV_s(f) = \max_{\circC:|\circC|\leq s}\ADV(\circC,f)$ and likewise for $\SUC_s(f)$. Let $\ADV^H,\SUC^H$ be the related notions when $x$ is sampled from the uniform distribution over $H\subset \zo^n$.
\end{definition}

We will repeatedly make use of an averaging sampler in order to make probabilistic method arguments randomness efficient. We first recall the definition of an averaging sampler, and then recall the classical result in \cite{RVW} that there exist highly efficient averaging samplers, even with exponentially small error.
\begin{definition}
    Given $m\in \N$ and $\eps,\delta>0$, we say that $\SAMP:\zo^l\ra (\zo^m)^t$ is a $t$-query $(m,\eps,\delta)$-\textbf{averaging sampler} with seed length $l$ if for every $g:\zo^m\ra [0,1]$ we have 
    \[
    \Pr_{q_1,\ldots,q_t\la \SAMP(U_l)}\left[\left|\E_{i\in [t]}[g(q_i)]-\E[g]\right|\leq \eps\right]\geq 1-\delta.
    \]
\end{definition}
\begin{proposition}[\cite{RVW}]\label{prop:samp}
    Given $m\in \N$ and $\eps>0$, there exists $t=\poly(m/\eps)$ and a $t$-query $(m,\eps ,2^{-2m})$-averaging sampler with seed length $4m$. Moreover, the sampler is evaluable in space $O(m)$.
\end{proposition}

Another tool that is repeatedly used in our proof is the combinatorial design, which is a family of subsets $S_1,\ldots,S_n\subseteq[s]$ such that $|S_i|=\alpha s$ for some constant $\alpha\in(0,1)$ and all $i\in[n]$, while $|S_i\cap S_j|\leq 2\alpha^2 s$ for all $i\neq j$.
The design will be used at two places: once in derandomized XOR Lemma (\Cref{sec:xor}) and once in the Nisan-Wigderson PRG (\Cref{sec:nw}). While the application in \Cref{sec:xor} only requires a linear-sized design, the application in \Cref{sec:nw} requires an exponential-sized design that is deterministically constructible in linear pace. The later was formally given in \cite{KM02}, so we concurrently use it for both applications.

\begin{proposition}[\cite{KM02}]\label{prop:design}
    For every $\alpha \in (0,1)$, there is $\beta \in (0,1)$ such that for $s\in \N$ one can deterministically generate in space $O(s)$ a combinatorial design of size $n=2^{\beta s}$ over $[s]$, that is, a family of subsets $S_1,\ldots,S_n\subseteq[s]$ such that $|S_i|=\alpha s$ and $|S_i\cap S_j|\leq 2\alpha^2 s$ for all $1\leq i<j\leq n$.
\end{proposition}

\subsection{Derandomizing the Polynomial Decoder}
For step (a) in \Cref{thm:IWrecon}, we need to convert a worst-case hard function to one with constant average-case hardness.

\begin{lemma}\label{lemma:RM}
    Given $f:\zo^m\ra\zo$, there is $g:\zo^{m'}\ra\zo$ where $m'=\Theta(m)$ such that, for every circuit $\circB$ such that $\SUC(\circB,g)>0.99$, there is a circuit $\circC$ of size $m^{O(1)}\cdot |\circB|$ such that 
    \[\circC(x)=f(x),\ \forall x\in\{0,1\}^m.\]
    Moreover, when $f$ is computable in space $O(m)$, $g$ is also computable in space $O(m)$, and there is a deterministic $O(m)$-space algorithm that, given the circuit $\circB$ which is evaluable in space $O(m)$, prints $\circC$, and $\circC$ is also evaluable in space $O(m)$.
\end{lemma}

The proof for \Cref{lemma:RM} is inspired by \cite{STV}, where we encode $f$ through Reed-Muller codes and switch to boolean domain via Hadamard codes. However, since we only need the resulting function to be average-case hard on a constant fraction of inputs, the code can be directly decoded instead of list-decoded, and we derandomize the decoding procedure with samplers. 

We need the following two facts. The first is a folklore fact on constructing low-degree extension, whose proof can be found at \cite[Proposition 2.2]{GKR15}:
\begin{proposition}\label{prop:RM}
    Given a finite field $\mathbb{F}$ and a subset $H\subseteq\mathbb{F}$, and oracle access to a function $f:H^\ell\ra\zo$, one can compute in space $O(\log|\mathbb{F}|+\log\ell)$ an $\ell$-variable polynomial $p:\mathbb{F}^\ell\ra\mathbb{F}$ that coincides with $f$ on $H^\ell$, and the degree of $p$ in each variable smaller than $|H|$.
\end{proposition}
The second fact concerns decoding Reed-Solomon codes:
\begin{proposition}\label{prop:BW}
    Given a finite field $\mathbb{F}$ with $|\mathbb{F}|=N$, whose elements can be canonically listed as $a_1,\ldots,a_N$ where $a_1=0$, there exists a circuit $\DEC:\mathbb{F}^N\ra\mathbb{F}^N$ that satisfies the following: If there exists a univariate polynomial $q:\mathbb{F}\ra\mathbb{F}$ of degree at most $d<N$, such that $q(a_i)=b_i$ for at least $(N+d)/2$ of $i\in[N]$, then
    \[\DEC(b_1,\ldots,b_N)=(q(a_1),\ldots,q(a_N)).\]
    Furthermore, $\DEC$ is of size $\poly(N)$ and depth $\polylog(N)$, and can be uniformly constructed in space $O(\log N)$ given the arithmetics in $\mathbb{F}$.
\end{proposition}
\begin{proof}
    The circuit $\DEC$ instantiates the Berlekamp-Welch algorithm \cite{WB86,GS92}. The algorithm involves solving systems of $O(N)$ linear equations on $O(N)$ variables, for which Csanky’s algorithm \cite{Csanky76} can be implemented in logspace-uniform-$\bm{\mathsf{NC}}$.
\end{proof}

\begin{proof}[Proof of \Cref{lemma:RM}]
    We assume without loss of generality that $m$ is a power of $2$. Let $\ell=m/\log m$, and $\mathbb{F}$ be a finite field of characteristic $2$ and size $m^2$. Take $H\subset\mathbb{F}$ to be a subset of size $m$, and we identify the domain $\{0,1\}^m$ of $f$ with $H^\ell$ as $2^m=|H|^\ell$. The arithmetics in $\mathbb{F}$ can be done in time $O(|\mathbb{F}|)$ and space $O(m)$, and so does the bijection between $\{0,1\}^m$ and $H^\ell$ (and its reverse).
    
    Let $p:\mathbb{F}^\ell\ra\mathbb{F}$ be the polynomial in \Cref{prop:RM}, and let
    $g:\mathbb{F}^{\ell+1}\ra\zo$ be the function defined as
    \[g(x_1,\ldots,x_\ell,y)=\langle p(x_1,\ldots,x_\ell), y\rangle,\]
    where $\langle\cdot,\cdot\rangle$ stands for inner product in $\mathbb{F}_2$ when taking the binary representation of the two arguments in $\mathbb{F}$. It is clear that $g$ can be computed in space $O(m)$, and the input of $g$ has length $(\ell+1)\log |\mathbb{F}|=O(m)$ when represented in binary.

    Now assume there is a circuit $\circB$ such that $\SUC(\circB,g)>0.99$. We first construct the circuit $\circB':\mathbb{F}^\ell\ra\mathbb{F}$ such that the $i$-th bit of the output is 
    \[\circB'_i(x_1,\ldots,x_\ell)=\MAJ_{z\in\mathbb{F}}(
    \circB(x_1,\ldots,x_\ell,e_i+z)-\circB(x_1,\ldots,x_\ell,z)).\]
    Here $e_i$ is the element in $\mathbb{F}$ whose binary representation has $1$ on the $i$-th bit and $0$ elsewhere. 
    
    \begin{claim}
        $\SUC(\circB',p)\geq 0.96$.
    \end{claim}
    \begin{proof}
        Since $\SUC(\circB,g)>0.99$, there are at least a $0.96$-fraction of $(x_1,\ldots,x_\ell)\in\mathbb{F}^\ell$ such that $\circB$ coincide with $g$ on more than $3/4$ of $y\in\mathbb{F}$, which contains both $z$ and $(e_i+z)$ with probability larger than $1/2$ for a random $z\in\mathbb{F}$. In such cases we have $\circB'_i(x_1,\ldots,x_\ell)=\langle p(x_1,\ldots,x_\ell), e_i\rangle$ for every $i$, and thus $\circB'(x_1,\ldots,x_\ell)=p(x_1,\ldots,x_\ell)$.
    \end{proof} 

    From $\circB'$, we reconstruct the circuit $\circC:\zo^m\ra\zo$ as follows. Let $\SAMP:\zo^{8m}\ra(\mathbb{F}^\ell)^t$ be the sampler in \Cref{prop:samp} with $\eps = 0.01$ and thus $t=\poly(m)$. We think of the $\SAMP$ as sampling $t$ random vectors $v=(v_1,\ldots,v_\ell)\in\mathbb{F}^\ell$, and given the input $x=(x_1,\ldots,x_\ell)\in H^\ell$ for $\circC$, each vector $v$ represents a line $\{x+\lambda v\mid\lambda\in\mathbb{F}\}$. On each line, $p(x+\lambda v)$ is a univariate polynomial on $\lambda$ of degree at most $\ell|H|=m^2/\log m$, and we use the decoder circuit $\DEC$ in \Cref{prop:BW} to decode the Reed-Solomon code given by $\circB'$ on the line. We let the value of $\circC(x)$ to be the most common (breaking ties arbitrarily) decoded value among the $t$ lines. Notice that this process depends on the seed of the sampler, and we actually go through all the seeds and choose the one that makes $\circC(x)$ correctly computes $f$ on all $x\in H^\ell$.
    
    Formally, we present this linear space reconstruction algorithm as \Cref{alg:RM}.

    \newcommand{\algRMRecon}{\textsc{RM\_RECON}}
    \begin{algorithm}\caption{\algRMRecon$(f,\circB)$}\label{alg:RM}
        \DontPrintSemicolon
        Let $\SAMP:\zo^{8m}\ra(\mathbb{F}^\ell)^t$ be the sampler of \Cref{prop:samp} with $\eps = 0.01$.\\
        \For{$y\in \zo^{8m}$}{
            Let $v_1,\ldots,v_t \la \SAMP(y)$.\\
            Let $\circC:\zo^m\ra\zo$ be the circuit
            \[\circC(x) = \MAJ_{i\in [t]}( \DEC_1( (\circB'(x+\lambda v_i))_{\lambda\in\mathbb{F}} ) ).\]\\
            \lIf{$\circC(x)=f(x)$ for all $x\in\{0,1\}^m$}{\Return{$\circC$}}
        }
    \end{algorithm}
    The circuit $\circC$ constructed in the algorithm is of size $2t|\mathbb{F}|^2|\circB|+m^{O(1)}=m^{O(1)}\cdot |\circB|$, and has additional depth $\polylog(m)$ compared to that of $\circB$. Therefore $\circC$ can be evaluated in space $O(m)$.

    Now we prove that the algorithm always returns a valid circuit $\circC$. Notice that for uniformly random $v\in\mathbb{F}^\ell$, $x+\lambda v$ is also uniformly random after given $x$ and $\lambda\neq 0$. Since $\SUC(\circB',p)\geq 0.96$, it means that there are at least a 0.84-fraction of $v\in\mathbb{F}^\ell$ such that $\circB'$ coincide with $p$ on $x+\lambda v$ for at least $3/4$ of $\lambda\in\mathbb{F},\lambda\neq 0$. Recall that the degree of $q(\lambda)=p(x+\lambda v)$ is at most $\ell|H|=|\mathbb{F}|/\log m$, and therefore by \Cref{prop:BW} we conclude that for every $x\in\{0,1\}^\ell$,
    \[\Pr_{v\in\mathbb{F}^\ell}\left[\DEC( (\circB'(x+\lambda v))_{\lambda\in\mathbb{F}} )= (p(x+\lambda v))_{\lambda\in\mathbb{F}}\right] \geq 0.84,\]
    in which case we have $\DEC_1( (\circB'(x+\lambda v))_{\lambda\in\mathbb{F}} )= p(x)$. Viewing this probability as an expectation of the indicator function on $v$, by the guarantee of the sampler in \Cref{prop:samp} we have
    \[\Pr_{v_1,\ldots,v_t\la \SAMP(U_{8m})}\left[
    \Pr_{i\in [t]}\left[\DEC_1( (\circB'(x+\lambda v_i))_{\lambda\in\mathbb{F}} )= p(x)\right] \geq 0.51
    \right]\geq 1-2^{-4m}.\]
    By a union bound over $x\in\{0,1\}^m$, there must exist a $y\in\{0,1\}^{8m}$ such that $\circC(x)=p(x)=f(x)$ for all $x\in\{0,1\}^m$. Therefore the algorithm always returns such a circuit $\circC$. Moreover, the algorithm can be implemented to run in space $O(m)$, as we can enumerate over seeds to the sampler and construct the circuit (as a function of the sampler output) in space $O(m)$, and test if the circuit correctly computes $f$ in this space bound. 
\end{proof}

\subsection{Derandomizing the Derandomized XOR Lemma}\label{sec:xor}

Our next step follows the approach of Impaggliazo and Wigderson~\cite{IW97}, who use a derandomized XOR lemma to produce from a function that is hard on a constant fraction of inputs, a function that is hard on any exponentially small fraction of inputs. The construction is identical to the one in \cite{IW97}, except that we modify the reconstruction algorithm and analysis to make the circuit $\circC$ constructible in deterministic space $O(m)$.
\begin{lemma}\label{lemma:derandXOR}
    For every $\gamma\in (0,1)$, there is an $O(m)$-space computable function $G:\zo^{m'}\ra (\zo^{m})^{m}$, where $m'=\Theta(m/\gamma)$, that satisfies the following: Given $f:\zo^m\ra\zo$, and a circuit $\circB$ satisfying $\SUC(\circB, f^m\circ G) \geq 2^{-\gamma m}$, there exists a circuit $\circC$ of size $2^{O(\gamma m)}\cdot |\circB|$ such that
    \[
    \SUC(\circC,f) >0.99.
    \]
    Moreover, when $f$ is computable in space $O(m)$, there is a deterministic $O(m)$-space algorithm that, given the circuit $\circB$ which is evaluable in space $O(m)$, prints $\circC$, and $\circC$ is also evaluable in space $O(m)$.
\end{lemma}

We first give the construction of the function $G$, which is called a \emph{direct-product generator} in \cite{IW97}. As in \cite{IW97}, it consists of two components: an expander walk and a combinatorial design. For the expander walk, we need an explicit expander where the neighbors of a vertex can be efficiently computed:
\begin{proposition}[see e.g. \cite{LPS88}]
    There is a constant $\lambda\in(0,1)$, such that for every $m\in\N$, there exists a $4$-regular graph $E_m$ on the vertex set $\zo^m$ with spectral expansion at most $\lambda$, such that given any vertex $v\in\zo^m$, its neighbors can be computed in time $\poly(m)$ and space $O(\log m)$.
\end{proposition}
Define the expander walk function $\mathtt{EW}:\zo^{3m}\ra(\zo^m)^m$ as follows: Given the input $v\in \zo^m$ and $d=(d_1,\ldots,d_m)\in[4]^m$, the output is sequence of vertices $v_1,\ldots,v_m$ in $E_m$ that starts with $v_1=v$, and take $v_{i+1}$ to be the $d_i$-th neighbor of $v_i$. On the other hand, let $S_1,\ldots,S_m\subseteq [s]$ be the first $m$ sets in the combinatorial design from \Cref{prop:design} with $\alpha=\gamma/2$ and $s=m/\alpha$. Then we defined the function $G:\zo^{3m+s}\ra(\zo^m)^m$ as:
\[G(r,v,d)=\big((r|_{S_1})\oplus \mathtt{EW}(v,d)_1, \ldots, (r|_{S_m})\oplus\mathtt{EW}(v,d)_m\big).\]
Here $r|_S$ is the part of $r\in\zo^s$ on indices $S$, and $\oplus$ is bit-wise XOR. From the definition we have that $G$ can be computed in time $\poly(m)$ and space $O(m)$. The input length of $G$ is $m'=3m+2m/\gamma=O(m/\gamma)$.

Now given $f:\zo^m\ra\zo$, assume there is a circuit $\circB$ such that $\SUC(\circB, f^m\circ G) \geq 2^{-\gamma m}$. Before we move on and show how to reconstruct the circuit $\circC$ efficiently and deterministically from $\circB$, let us first review the reconstruction step in \cite{IW97}. For $i\in[m]$, $x\in\zo^m$, $a\in\zo^{s-m}$, $v\in\zo^m$ and $d\in[4]^m$, let $h(i,x,a,v,d)=(r,v,d)$ where $r\in\zo^s$ such that
\[r|_{S_i}=x\oplus\mathtt{EW}(v,d)_i
\textrm{ and }r|_{\overline{S_i}}=a.\]
The function $h$ is called the \emph{restricting function} of $G$. Given $x\in\zo^m$, with $i,a,v$ and $d$ chosen uniformly at random, they build a circuit $\circF$ that first simulates $\circB$ to compute $\circB(h(i,x,a,v,d))=(y_1,\ldots,y_m)$. Then it computes a number $t$ defined as
\[t=\left|\left\{j\neq i\mid y_j\neq f\big(G_j\circ h(i,x,a,v,d)\big)\right\}\right|,\]
and outputs $y_i$ with probability $2^{-t}$, while outputting a random bit with probability $1-2^{-t}$. To compute $t$, for each $j\neq i$, $f(G_j\circ h(i,x,a,v,d))$ is computed through a non-uniformly constructed look-up table for $f$ of size $2^{\gamma m}$, containing the values of $f(x_j)$ for all possible $j$-th output $x_j$ of $G\circ h$ with the fixed $i,a,v$ and $d$.

We could not resort to non-uniformity to construct the look-up table. Nevertheless, when $f$ is computable in space $O(m)$, we can compute the entire table in space $O(m)$ and hardwire it to the circuit. Even better, when $i,a,v$ and $d$ are given, each output $x_j$ of $G\circ h$ is fixed except for $\gamma m$ bits (corresponding to the coordinates in $S_i\cap S_j$), so we only need to go through all $2^{\gamma m}$ possibilities for these bits to compute the table.

The circuit $\circF$ presented above uses a string $R$ of $|R|=O(m)$ random bits, including $i,a,v,d$ along with $w\in\zo^{m+1}$, the randomness used to decide the final output. It was proved in \cite{IW97} that:
\begin{proposition}[{\cite[Theorem 15]{IW97}}]\label{prop:iw}
    Suppose the $\SUC(\circB, f^m\circ G) \geq 2^{-\gamma m}$. There exists $c>0$ (that depends on $\gamma$), such that the fraction of inputs $x\in \zo^m$ with \[\Pr_R[\circF(x,R)=f(x)]\geq 1/2+2^{-\gamma m}/c\]
    is more than 0.99.
\end{proposition}
Therefore, the final circuit $\circC$ takes $O(m\cdot 2^{2\gamma m})$ independent copies of $\circF$ and outputs their majority, and there exists a fixing of the randomness that provides the final deterministic circuit $\circC$. We could not afford to store exponentially many random bits if they are independently sampled. Instead, we employ the efficient sampler in \Cref{prop:samp} that uses only $O(m)$ random bits as the seed to generate $2^{O(\gamma m)}$ samples, and we can enumerate over all the seeds to find the one that makes $\SUC(\circC,f)>0.99$. As shown in the proof below, such seed always exists.

\begin{proof}[Proof of \Cref{lemma:derandXOR}]
    Let $\circF:\zo^{m+|R|}\ra\zo$ be the circuit described above, and $c>0$ be the constant in \Cref{prop:iw}. We give the formal description of the linear-space algorithm for the reconstruction procedure as \Cref{alg:xor}.
    \newcommand{\algXORRecon}{\textsc{XOR\_RECON}}
    \begin{algorithm}\caption{\algXORRecon$(f,\circB)$}\label{alg:xor}
        \DontPrintSemicolon
        Let $\SAMP:\zo^{4|R|}\ra (\zo^{|R|})^t$ be the sampler of \Cref{prop:samp} with $\eps=2^{-\gamma m}/(2c)$.\\
        \For{$y\in \zo^{4|R|}$}{
            Let $R_1,\ldots,R_t \gets \SAMP(y)$.\\
            Let $\circC:\zo^m\ra\zo$ be the circuit
            \[\circC(x) = \MAJ(\circF(x,R_1),\ldots,\circF(x,R_t)).\] \\
            \lIf{$\SUC(\circC,f)>0.99$}{\Return{$\circC$}}
        }
    \end{algorithm}

    By \Cref{prop:samp} we have $t=\poly(m/\eps)=2^{O(\gamma m)}$ for $\eps=2^{-\gamma m}/(2c)$.
    From the description we know that $\circF$ has size $|\circB|+2^{\gamma m}\cdot m^{O(1)}$, and therefore $\circC$ has size $t|\circF|+m^{O(1)}=2^{O(\gamma m)}\cdot |\circB|$. When $\circB$ is evaluable in space $O(m)$, $\circC$ is clearly also evaluable in space $O(m)$.

    By the guarantee of the averaging sampler $\SAMP$ in \Cref{prop:samp}, for every $x\in\zo^m$:
    \[
        \Pr_{R_1,\ldots,R_t\la \SAMP(U_{4|R|})}\left[\left|\E_{i\in [t]}[\circF(x,R_i)]-\E_R[\circF(x,R)]\right|\leq \eps\right]\geq 1-2^{-2|R|}.
    \]
    By \Cref{prop:iw}, there exists a subset $V\subseteq\zo^m$ such that $|V|>0.99\cdot 2^m$, such that for every $x\in V$:
    \[
        \left|\E_R[\circF(x,R)]-f(x)\right|\leq 1/2-2^{-\gamma m}/c=1/2-2\eps.
    \]
    Therefore for every $x\in V$, it is implied that
    \[
        \Pr_{R_1,\ldots,R_t\la \SAMP(U_{4|R|})}\left[\left|\E_{i\in [t]}[\circF(x,R_i)]-f(x)\right|\leq 1/2-2\eps+\eps\right]\geq 1-2^{-2|R|},
    \]
    which means that
    \[
        \Pr_{R_1,\ldots,R_t\la \SAMP(U_{4|R|})}\left[\MAJ(\circF(x,R_1),\ldots,\circF(x,R_t))= f(x)\right]\geq 1-2^{-2|R|}>1-1/|V|.
    \]
    By a union bound over $x\in V$, there must exist a $y\in\{0,1\}^{4|R|}$ such that $\circC(x)=f(x)$ for all $x\in V$, which satisfies $\SUC(\circC,f)>0.99$. Therefore the algorithm always returns a valid $\circC$. Moreover, the algorithm runs in space $O(m)$, as it enumerates the seeds of length $O(|R|)=O(m)$, constructs and evaluates the circuit $\circC$ and makes oracle calls to $f$, which all can be done in space $O(m)$.
\end{proof}

\subsection{Derandomizing the Goldreich-Levin Theorem}
\begin{lemma}\label{lemma:GL}
    Given $f:\zo^m\ra\zo^m$, let $g:\zo^m\times\zo^m\ra\zo$ be defined as $g(x,r)=\langle f(x),r\rangle$. Then, given $\delta>0$, there is $\delta' \geq \Omega(\delta^3/m)$ so that, for every $\circB$ satisfying $\ADV(\circB,g)> \delta$, there is a circuit $\circC$ of size at most $|\circB|\cdot (m/\delta)^{O(1)}$ satisfying
    \[\SUC(\circC,f)>\delta'.\]
    Moreover, when $f$ is computable in space $O(m)$, there is a deterministic $O(m)$-space algorithm that, given the circuit $\circB$ which is evaluable in space $O(m)$, prints $\circC$, and $\circC$ is also evaluable in space $O(m)$.
\end{lemma}

Note that the original Goldreich-Levin theorem~\cite{GL89} does not guarantee (and in fact does not give) an efficient deterministic reconstructor, as it is not randomness efficient. A later work of Hoza and Klivans~\cite{HK18} achieves this, though with a significantly more involved proof. As such, we directly show this using small-bias spaces, which we define now:
\begin{definition}\label{def:smallbias}
    A function $G:\zo^t\ra\zo^k$ is an $\eps$-\textbf{biased generator} if $G(U_t)$ is a $\eps$-biased probability space over $\zo^k$, which formally means that for every $T\in \zo^k$,
    \[\Pr_{y\la U_t}\left[\langle T, G(y)\rangle=1\right] \in [1/2-\eps,1/2+\eps].\]
\end{definition}
We recall small-bias generators exist with good seed length, and moreover these generators can be evaluated in small space:
\begin{proposition}[\cite{NaorN93}]\label{prop:smallbias}
    Given $k\in \N$ and $\eps>0$, there is an $O(t)$-space evaluable $\eps$-biased generator $\BIAS:\zo^t\ra\zo^k$ with seed length $t=O(\log(k/\eps))$.
\end{proposition}

We require a basic Fourier-analytic lemma, that states that a small-bias space fools the conjunction of $k$ parities.
\begin{lemma}\label{lem:Fourier}
    Let $\BIAS:\zo^t\ra\zo^k$ be an $\eps$-biased generator. Then for every collection $T_1,\ldots,T_d\in \zo^k$ and $v_1,\ldots,v_d \in \zo$ we have
    \[
    \left|\E_{r\la \BIAS(U_t)}\left[\bigwedge_{i\in [d]}(\langle T_i,r\rangle \oplus v_i )\right]-\E_{r\la U_k}\left[\bigwedge_{i\in [d]}(\langle T_i,r\rangle \oplus v_i)\right]\right|\leq 2\eps.
    \]
\end{lemma}
\begin{proof}
    We have
    \begin{align*}
        \bigwedge_{i\in [d]}(\langle T_i,r\rangle \oplus v_i ) &= 1-2\cdot 2^{-d}\sum_{S\subseteq[d]}\bigoplus_{i\in S}\neg\left(\langle T_i,ry\rangle \oplus v_i\right)\\
        &= 1-2\cdot 2^{-d}\sum_{S\subseteq [d]}\left(\left\langle \bigoplus_{i\in S}T_i,r\right\rangle\oplus \bigoplus_{i\in S}\neg v_i\right)
    \end{align*}
    and as $\BIAS$ fools all such parities to error $\eps$ in the summation over $S\subseteq[d]$, we have that the total error is at most $2\eps$.
\end{proof}

\begin{proof}[Proof of \Cref{lemma:GL}]
If $\delta<2^{-m}$, we can choose $\delta'=2^{-m}$ and the lemma trivially holds for a circuit $\circC$ outputting a constant. Therefore, from now on we assume that $\delta\geq 2^{-m}$. We formally state our algorithm as \Cref{alg:GL}, with $\delta'$ to be determined later. Note that $\ell=O(m)$, and therefore in the $\eps$-biased generator $\BIAS:\zo^t\ra \zo^{\ell\times m}$ we have $t=O(\log(\ell m/\eps))=O(m)$ with $\eps=2^{-4m-1}$, and the algorithm runs in space $O(t+\ell+m)=O(m)$.

\newcommand{\algGLRecon}{\textsc{GL\_RECON}}
\begin{algorithm}\caption{\algGLRecon$(f,B)$}\label{alg:GL}
    \DontPrintSemicolon
    Let  $\ell \gets \lceil \log_2(128m/\delta^2+1)\rceil$.\\
    Let $\BIAS:\zo^{t}\ra \zo^{\ell\times m}$ be the generator of \Cref{prop:smallbias} with $\eps=2^{-4m-1}$.\\
    \For{$y\in \zo^{t}$}{
        Let $r_1,\ldots,r_\ell \la \BIAS(y)$.\\
        \For{$(b_1,\ldots,b_\ell)\in \zo^\ell$}{
            Let $\circC:\zo^m\ra\zo^m$ be the circuit that for each $i\in[m]$:
            \[\circC_i(x) = \MAJ_{J\subseteq[\ell]:J\neq \emptyset}(b^J\oplus \circB(x,r^J\oplus e_i)).\] \\
            \lIf{$\SUC(\circC,f)>\delta'$}{\Return{$\circC$}.}
        }
    }
\end{algorithm}
    
    We view the output of $\BIAS$ as a tuple of $\ell$ vectors:
    \[
    \BIAS(y) = (r_1,\ldots,r_\ell), \quad r_i\in \zo^m.
    \]
    For convenience, let $\vec{r}:=(r_1,\ldots,r_\ell)$ and $\vec{b}:=(b_1,\ldots,b_\ell)$. For every $J\subseteq [\ell]$, let:
    \[r^J :=\bigoplus_{i\in J}r_i,\quad\quad b^J = \bigoplus_{i\in J} b_i.\]
    Note that in the original GL algorithm, all $r_i$'s are i.i.d. uniform over $\zo^m$. We first argue that our distribution over $r^J$'s satisfies (approximately) the two properties used in the analysis of the original algorithm:
    \begin{claim}\label{clm:closeUnif}
        The following two properties hold:
        \begin{enumerate}
            \item For every non-empty $J$, $r^J$ is $2^{-2m}$-close to $U_m$ in $\ell_1$-distance.
            \item For every non-empty $J$ and $J'$ where $J\neq J'$, $(r^J,r^{J'})$ is $2^{-2m}$ close to $U_{2m}$ in $\ell_1$-distance.
        \end{enumerate}
    \end{claim}
    \begin{proof}
        For $i\in [m]$, the $i$-th bit of $r^J$ can be written as $\langle T_{i,J},\BIAS(y)\rangle$ where $T_{i,J}$ indicates a non-empty subset of bits. From \Cref{lemma:GL} we know that for every $v\in\zo^m$,
        \[\left|\Pr_{y\la U_t}[r^J=v]-
        \Pr_{r\la U_{\ell m}}\left[\bigwedge_{i\in[m]}(\langle T_{i,J},r\rangle=v_i)\right]\right|\leq 2\eps.\]
        Notice that $\{T_{i,J}\}_{i\in[m]}$ are linearly independent, and thus $(\langle T_{i,J},r\rangle)_{i\in[m]}$ is uniformly distributed over $\zo^m$. Therefore taking the sum over $v\in\zo^m$ we have that $r^J$ is $2\eps\cdot 2^m\leq 2^{-2m}$-close to $U_m$ in $\ell_1$ distance.

        When $J\neq J'$ are both non-empty, $\{T_{i,J}\}_{i\in[m]}\cup \{T_{i,J'}\}_{i\in[m]}$ are still linearly independent. For the same reason as above, $(r^J,r^{J'})$ is $2\eps\cdot 2^2m=2^{-2m}$-close to $U_{2m}$ in $\ell_1$ distance.
    \end{proof}
    
    Now recall that for $i\in [m]$ the $i$th bit of the output of $\circC$ is 
    \[
    \circC_i(x) = \MAJ_{J:J\neq \emptyset}(b^J\oplus \circB(x,r^J\oplus e_i)).
    \]
    Thus $\circC$ has size $|\circC| \leq (|\circB|+O(\ell))\cdot 2^\ell m\leq |\circB|\cdot O(2^\ell \ell m)=|\circB|\cdot (m/\delta)^{O(1)}$ as claimed. To analyze the performance of $\circC$, let
    \[S := \{x\in \zo^m:\Pr_{z\la U_m}[\circB(x,z)=g(x,z)]\geq 1/2+\delta/2\}.\]
    By a standard averaging argument, $|S| \geq (\delta/2)\cdot 2^m$.
    \begin{claim}\label{claim:majbias}
        For every $x\in S$ and $i\in [m]$,
        \[
        \Pr_{(r_1,\ldots,r_\ell)\la \BIAS(U_t)}\left[\left|\{J:\circB(x,r^J\oplus e_i) = g(x,r^J\oplus e_i)\}\right|\leq\frac{1}{2}(2^\ell-1)\right]\leq \frac{1}{2m}.
        \]
    \end{claim}
    \begin{proof}
        For the remainder of the proof we fix $x$ and $i$.
        Let $A\subset \zo^m$ be the set of values $r$ on which $\circB(x,r) = g(x,r)$. By the fact that $x\in S$ we have $|A|\geq (1/2+\delta/2)\cdot 2^m$. Furthermore, for each $y\in\zo^t$ (where $y$ is the input to $\BIAS$) let 
        \[\zeta_J(y) = \mathbb{I}[r^J\oplus e_i \in A]\]
        and observe that $\zeta_J=1$ is equivalent to $\circB(x,r^J\oplus e_i)= g(x,r^J\oplus e_i)$, i.e. $\circB$ computes the inner product with $f(x)$ correctly on that input. Now observe that by \Cref{clm:closeUnif}, 
        \[\E_y[\zeta_J]=\Pr_y[\zeta_J(y)=1]\geq 1/2+\delta/2 - 2^{-2m} \geq 1/2+\delta/4.\]
        We now bound the variance of the number of such places where we compute the inner product correctly. Let
        \begin{align*}
            \sigma^2 = \Var\left(\sum_{J}\zeta_J\right) 
            &= \sum_{J,J'}\Cov(\zeta_J,\zeta_{J'})\\
            &\leq \sum_J\Var(\zeta_J)+\sum_{J,J'} 2^{-2m}\\
            &\leq  2^\ell +2^{2\ell}\cdot 2^{-2m} \leq 2^{\ell+1}
        \end{align*}
        where the first inequality follows from \Cref{clm:closeUnif}.
        Now the result follows by Chebyshev's inequality and a union bound. For convenience let $d=2^\ell-1$, and the probability in the claim equals:
        \begin{align*}
            \Pr_y\left[\sum_J\zeta_J \leq \frac{d}{2}\right] &\leq \Pr_y\left[\left|\sum_J\zeta_J - \E[\zeta_J]\cdot d\right| \geq \left(\frac{d \delta}{4 \sigma}\right)\cdot \sigma\right]\\
            &\leq \frac{16\sigma^2}{\delta^2 (2^\ell-1)^2}
            \leq \frac{32\sigma^2}{\delta^2 2^{2\ell}}
            \leq\frac{64}{\delta^2 2^{\ell}} \leq \frac{1}{2m}. \qedhere
        \end{align*}
    \end{proof}
    Notice that when $\circB(x,r^J\oplus e_i) = g(x,r^J\oplus e_i)$ and for every $j\in[\ell]$, $b_j=g(x,r_j)$, we have
    \[b^J\oplus \circB(x,r^J\oplus e_i)=g(x,r^J)\oplus g(x,r^J\oplus e_i)=g(x,e_i)=f_i(x).\]
    Thus, using a union bound over $i\in[m]$ on \Cref{claim:majbias}, we have that for every $x\in S$,
    \begin{align*}
        \Pr_{\substack{\vec{r}\la \BIAS(U_t)\\ \vec{b}\la U_l}}[\circC(x)=f(x)] 
        \geq &\ \Pr_{\vec{r}\la \BIAS(U_t)}\left[\forall i\in[m],\left|\{J:\circB(x,r^J\oplus e_i) = g(x,r^J\oplus e_i)\}\right|>\frac{1}{2}(2^l-1)\right]\\
        &\ \cdot\Pr_{\vec{b}\la U_\ell}\big[\forall j\in [\ell], b_j = g(x,r_j)\big]\\
        \geq &\ \frac{1}{2}\cdot\Pr_{\vec{b}\la U_\ell}\big[\forall j\in [\ell], b_j = g(x,r_j)\big] \geq 2^{-\ell-1}.
    \end{align*}
    Thus, there is an assignment of $y$ and $\vec{b}$ such that $\circC$ computes $f$ correctly on at least $|S|\cdot 2^{-\ell-1} \geq 2^m\cdot \delta 2^{-\ell-2}$ inputs. Moreover, we can find such a circuit by enumerating the assignments to $y$ and $\vec{b}$, and verifying the success probability by evaluating $\circC$ and $f$ over all $x\in\zo^m$. Therefore letting
    \[\delta'=\delta 2^{-\ell-2}=\Omega(\delta^3/m)\]
    completes the proof.
\end{proof}

\subsection{Space-Efficient Nisan-Wigderson PRG}\label{sec:nw}
We recall the argument of~\cite{KM02} that there is a space-efficient implementation of the Nisan-Wigderson~\cite{NW94} PRG, using the linear-space constructible combinatorial design (\Cref{prop:design}). While we rephrase their result in our notation, we make no changes to the construction, as (in contrast to all other steps) the existing implementation satisfies our desired reconstruction property.

\begin{lemma}\label{lemma:NW}
    Given $\rho>0$ and $n\in \N$ and a family of functions  $f_m:\zo^m\ra\zo\in\SPACE[m]$, there exists an $m=\Theta(\log n)$ and $G:\zo^s\ra\zo^n$ with $s=O(m)$ such that, given a circuit $\circB$ which is a next-bit predictor for $G$ with advantage $\varepsilon$, there is a circuit $\circC$ of size $|\circB|+O(n2^{\rho m})$ satisfying
    \[\ADV(\circC,f_m)>\eps.\]
    Moreover, there is an deterministic $O(m)$-space algorithm that, given the circuit $\circB$ which is evaluable in space $O(m)$, prints $\circC$, and $\circC$ is also evaluable in space $O(m)$.
\end{lemma}
\begin{proof}[Proof of \Cref{lemma:NW}]
    Fix $\alpha\in(0,1)$ such that $\alpha\leq\rho/2$, and let $\beta\in(0,1)$ be the constant in \Cref{prop:design}. Choose $s=O(\log n)$ such that $2^{\beta s}= n$, and let $m=\alpha s$.
    Let $\mathcal{S}=(S_1,\ldots,S_n)$ be the design of \Cref{prop:design} over $[s]$ with parameter $\alpha$, and let $f_m:\zo^m\ra\zo$ be the function on inputs of size $m=O(\log n)$.
    
    We let $G(x) := f(x_{S_1})f(x_{S_2})\ldots f(x_{S_n}).$
    Now suppose $\circB$ is an $\eps$-next-bit predictor for bit $i$ of $G$, i.e.
    \[
    \Pr_{x
    \la U_s}[\circB(G(x)_{1..i})=G(x)_{i+1}]>\frac{1}{2}+\eps.
    \]
    Then let $S:=S_{i+1}$ and $T:=[s]\setminus S_{i+1}$ and write the above inequality as
    \[
    \Pr_{(x_S,x_T)\la U_s}[\circB(G(x_S \cup x_T)_{1..i})=f(x_S)]>\frac{1}{2}+\eps.
    \]
    For each fixing of $x_T$, we let the circuit $\circC$ to be $\circC(x_S)=\circB(G(x_S \cup x_T)_{1..i})$. Then we have 
    \[\E_{x_T}[\ADV(\circC,f_m)]>\eps.\]
    Thus, the algorithm can enumerate over all possible assignments to $x_T$ in space $|T|=O(m)$, and for each assignment check the advantage of $\circC$. Once the algorithm has found the fixing of $x_T$ such that the restricted circuit has advantage at least $\eps$, for every $j\leq i$, the $j$-th bit of the output of $G(x_S \cup x_T)$, which is $f(x_{S_j})$, depends on $|S\cap S_j|\leq 2\alpha^2s=\rho m$ bits of $x_S$, and hence we can output a ($O(m)$-space constructible) circuit for $f(x_{S_j})$ size at most $O(2^{\rho m})$, and hence the total size of $\circC$ is at most $|\circB|+O(n2^{\rho m})$. 
\end{proof}

\subsection{Putting It All Together}\label{subsec:combine}

\begin{proof}[Proof of~\Cref{thm:IWrecon}]

    Given $\eps$, we first do the construction steps. For each $m\in\mathbb{N}$:
    \begin{enumerate}
        \item Let $f':\zo^{m_1}\ra\zo$ be the function $g$ of \Cref{lemma:RM} applied to $f_m$.
        \item Let $f'':\zo^{m_2}\ra \zo^{m_1}$ be the function $f'^{m_1}\circ G$ of \Cref{lemma:derandXOR} applied to $f'_{m_1}$ with the constant $\gamma$ to be chosen later.
        \item Let $f''':\zo^{m_3}\ra\zo$ be the function $g$ of \Cref{lemma:GL} applied to $f''_{m_2}$ with the constant $\delta$ to be chosen later.
        \item Let $G:\zo^s\ra\zo^n$ be the function of \Cref{lemma:NW} applied to $f'''_{m_3}$ and $\circB$ with the constant $\rho$ to be chosen later.
    \end{enumerate}
    Notice that $m_1,m_2,m_3$ and $s$ are all $\Theta(m)$, and the functions $f',f'',f'''$ and $G$ are all computable in space $O(m)$.
    
    Suppose now we are given a $1/(8n)$ next-bit predictor $\circB$ for $G$ of size $n^2$.
    As $n$ is given, we decide the value of $m_3=\Theta(\log n)$ through \Cref{lemma:NW}, which in turn decides the value of $m=\Theta(\log n)$. The reconstruction steps go as follows:
    \begin{enumerate}
        \item[4.] By \Cref{lemma:NW}, we can construct in space $O(m)$ a circuit $\circC_3$ such that $\ADV(\circC_3,f'''_{m_3})>1/(8n)$, and $\circC_3$ has size $s_3=n^2+O(n2^{\rho m_3})\leq 2^{c_3\rho m}$ for some constant $c_3>0$.
        \item[3.] By \Cref{lemma:GL}, where we now set $\delta=1/(8n)$, we can construct in space $O(m)$ a circuit $\circC_2$ such that $\SUC(\circC_2,f''_{m_2})>\Omega(\delta^3/m_2)\geq 2^{-c_2 \rho m}$, and $\circC_2$ has size $s_2=s_3\cdot (m_2/\delta)^{O(1)}\leq 2^{c_2\rho m}$ for some constant $c_2>0$.
        \item[2.] By \Cref{lemma:derandXOR}, where we now set $\gamma=c_2\rho $, we can construct in space $O(m)$ a circuit $\circC_1$ such that $\SUC(\circC_1,f'_{m_1})>0.99$ and $\circC_1$ has size $s_1=s_2\cdot 2^{O(\gamma m_1)}\leq  2^{c_1 \rho m}$ for some constant $c_1>0$. 
        \item[1.] By \Cref{lemma:RM}, we can construct in space $O(m)$ a circuit $\circC$ such that $\circC(x)=f_m(x)$ for every $x\in\zo^m$, and $\circC$ has size $s=s_1\cdot m^{O(1)}\leq 2^{c_0\rho m}$ for some constant $c_0>0$. By choosing $\rho=\eps/c_0$, we obtain the final result. \qedhere
    \end{enumerate}
\end{proof}

\section{Universal Derandomization of BPL} \label{section:universal}
\newcommand{\algUniv}{\textsc{UnivDerand}}
\newcommand{\algUnivNL}{\textsc{UnivNL}}
\newcommand{\algLCTest}{\textsc{LCTest}}
\newcommand{\algTarget}{\textsc{TargetPRG}}
\newcommand{\pv}{p_{\ra v}}
\newcommand{\tpv}{\widetilde{p_{\ra v}}}

Here we state the main theorem of this section, that there exists a universal derandomizer for logspace computation.
\begin{theorem}\label{thm:universal}
    There is a deterministic machine $\algUniv$ such that:
    \begin{itemize}
        \item\label{itm:univGood} On input $1^n$ and an OBP $B$ of length and width at most $n$, outputs $\delta:=\algUniv(1^n,B)$ satisfying $|\delta-\E[B]|<n^{-1}$.
        \item\label{itm:univIFF} For every space-constructible function $S:\N\ra \N$ satisfying $S(n)\geq \log n$, $\algUniv$ runs in space $O(S(n))$ if and only if $\prBPL\subseteq \SPACE[O(S(n))]$.
    \end{itemize}
\end{theorem}
We first give the intuitive explanation of the algorithm executed by the machine $\algUniv$. It enumerates over Turing machines $\langle i \rangle$ and space bounds $j$. At each step, $\algUniv$ runs $\langle i \rangle$ on input $(1^n,B_{\ra v})$ for every $v$, where $\Bv$ for $v\in V_i$ is the program that is identical to $B$ in the first $i$ layers, then accepts if the program reaches state $v$. If $\langle i\rangle$ ever touches more than $j$ spaces on the work tape, $\algUniv$ halts and increments $i$ or $j$. Otherwise, we have a set of estimates $\{\tpv\}:=\{\langle i\rangle(1^n,\Bv)\}$ (and note we can generate these estimates on the fly in space $O(j+\log n)$). We then submit these estimates to the local consistency test of Cheng and Hoza~\cite{CH20}, and if the test passes, we return the estimate of the probability of reaching the accepting state.
\begin{theorem}[\cite{CH20}]\label{thm:LCprobs}
    There is a deterministic logspace algorithm \algLCTest~ that takes as input $1^n$ and an OBP $B$ with length and width at most $n$ and the estimates $\{\tpv\}_{v\in V}$. If for every $v$, $|\tpv-\pv|\leq n^{-3}$, the algorithm accepts, and moreover if the algorithm accepts, $|\tpv-\pv|\leq n^{-1}$ for every $v$. 
\end{theorem}
Note that the true probabilities $\pv$ only appear in the statement of the theorem, and are not part of the input to the testing algorithm.
We can now give the formal description of the algorithm as \Cref{alg:univ}. By soundness of the test \algLCTest, if $\algUniv$ returns a value, the value must be a good approximation of the acceptance probability, so it suffices to show this occurs (and occurs in the desired space bound).

\begin{algorithm}
\caption{\algUniv$(1^n,B)$}\label{alg:univ}
    \DontPrintSemicolon
    \SetKwFor{Whenever}{whenever}{do}{end}
    \For{$j\gets0,1\ldots,$}{
        \For{$i\gets 0,1,\ldots,j$}{
            \For{$r \gets 1\cdot n^{-5}/2,2\cdot n^{-5}/2,\ldots,2n^2\cdot n^{-5}/2$}{
                Compute $b\gets \algLCTest(1^n,B,\{\langle i\rangle(1^n,\Bv,r)\}_{v\in V(B)})$; \\
                \Whenever{$\langle i\rangle$ uses more than $j$ space or more than $2^j$ time}{
                Abort the simulation of $\langle i\rangle$ and pass to the next $r$.
                }
                \lIf{$b=1$}{\Return{$\langle i\rangle(1^n,B,r)$.}}
            }
        }
    }
\end{algorithm}

To do so, we rely on the promise search problem $f_c$ with parameter $c\in\N$ (which we define as a function outputting a value in $[0,1]$ for convenience) defined as follows. Given $1^n$, an ordered branching program $B$ of length and width at most $n$, and a rounding threshold $r$ with the promise that 
\[\left|\E[B]-k\cdot n^{-c+2}+r\right|>\frac{n^{-c}}{6},\ \forall k\in\Z,\]
i.e. $\E[B]+r$ is polynomially bounded away from every multiple of $n^{-c+2}$, the problem asks to output a (pseudo-deterministic) number $f_c(1^n,B,r)$ that is within $n^{-c+2}$ distance of $\E[B]$. The presence of the rounding value, inspired by the approach of Saks and Zhou~\cite{SZ99}, is because when $\E[B]$ are very close to a threshold, it becomes hard to determining whether the expectation is above or below the cutoff.

We prove in \Cref{prop:BPLcomplete} that the task of computing $f_c$ is $\promBPL$ complete for every $c\in\N$. Therefore, if $\prBPL\subseteq\SPACE[S(n)]$, there is a machine $\langle i\rangle$ that computes $f_c$ in space $j=O(S(n))$. Finally, to accommodate the presence of the rounding threshold, $\algUniv$ additionally enumerates over a polynomial number of choices for $r$. We show that there exists a proper $c\in\N$ such that for every $B$, a good $r$ that satisfies the promise of $f_c$ exists. This is essentially proved via the argument of Saks and Zhou~\cite{SZ99}. Hence, the algorithm will always find a tuple $(i,j,r)$ such that we obtain good estimates of $\E[\Bv]$ for every $v$, and thus the machine will halt and return the correct value.

\begin{proposition}\label{prop:BPLcomplete}
    For every $c\in \N$, let $f_c$ be the problem where, given $1^n$ and an ordered branching program $B$ of length and width at most $n$, and $r\in [0,1]$ such that for every $k\in \Z$, 
    \[\left|\E[B]-k\cdot n^{-c+2}+r\right|>\frac{n^{-c}}{6},\]
    return with probability at least $2/3$ the same number $\delta$ that satisfies $\left|\E[B]-\delta\right|\leq n^{-c+2}$. Then $f_c$ is $\prBPL$-complete under $\L$ reductions.
\end{proposition}
\begin{proof}[Proof Sketch]
    Fix arbitrary $c\in \N$.
    We first prove $f_c\in \prBPL$. Let $\mathcal{R}(1^n,B,r)$ be an algorithm that takes $n^{2c+1}$ random walks from $\vs$ over $B$, and let $\gamma$ be the fraction of these walks which reach $\va$. Let $k\in\Z$ be the largest value such that $\gamma+r \geq k\cdot n^{-c+2}$, and return $\delta=k\cdot n^{-c+2}$. Since this algorithm clearly runs in randomized logspace, it suffices to show that, for $B$ and $r$ that satisfy the promise, there is some fixed $k$ that $\mathcal{R}$ identifies with probability over $2/3$. Note that by the promise, we have that for some $k_0\in\Z$,
    \[k_0\cdot n^{-c+2}+\frac{n^{-c}}{6} < \E[B]+r< (k_0+1)\cdot n^{-c+2}-\frac{n^{-c}}{6}. \]
    On the other hand, using concentration bounds we can show that with probability at least $2/3$,
    \[\left|(\E[B]+r)-(\gamma+r)\right|=\left|\E[B]-\gamma\right|\leq \frac{n^{-c}}{6}.\]
    In this case $\mathcal{R}$ always identifies $k=k_0$ since $k_0\cdot n^{-c+2}<\gamma+r<(k_0+1)\cdot n^{-c+2}$.

    We now prove that $f_c$ is $\prBPL$-hard. We recall the standard $\prBPL$-complete problem: Given an OBP $B$ of length and width $n$, determine if $\E[B]<1/3$ or $\E[B]>2/3$, where the promise is that one of these cases holds. We reduce this problem to $f_c$ as follows. Let $T_B:\zo^{dn}\ra\zo$ be the OBP defined as
    \[
    T_B(x_1,\ldots,x_d) = \MAJ(B(x_1),\ldots,B(x_d))
    \]
    where $d=O(c\log n)$ such that if $\E[B]<1/3$ then $\E[T_B]<n^{-c}/6$, and if $\E[B]>2/3$ then $\E[T_B]>1-n^{-c}/6$. Observe that $T_B$ has length and width $N=\poly(n)$ and is constructible in deterministic logspace given $B$. Thus, let the input to $f_c$ be $(1^N,T_B,n^{-c})$, which satisfies the promise of $f_c$, and hence if the answer is less than $1/2$ we determine that $\E[B]<1/3$, and otherwise determine that $\E[B]>2/3$.
\end{proof}

We first prove that the values the machine returns are accurate (assuming the machine returns a value).
\begin{lemma}\label{lem:decides}
    For every $B$, $|\algUniv(1^n,B)-\E[B]|\leq n^{-1}$.
\end{lemma}
\begin{proof}
    This follows from \Cref{thm:LCprobs} applied to $\tpv=\langle i\rangle(1^n,\Bv,r)$.
\end{proof}
We next prove the machine halts in the claimed space bound.
\begin{lemma}\label{lem:halts}
    For every space-constructible function $S:\N\ra \N$ with $S(n)\geq \log n$, $\algUniv$ runs in space $O(S(n))$ if $\prBPL\subseteq \SPACE[O(S(n))]$.
\end{lemma}
\begin{proof}
    We prove that $\algUniv(1^n,B)$ halts and returns a value with $i+j\leq c\cdot S(n)$ for an absolute constant $c$ (in particular, $i,j<\infty$), which suffices to establish the lemma by the composition of space-bounded algorithms.
    
    By~\Cref{prop:BPLcomplete}, there is a Turing machine $\langle i \rangle$ deciding the language $f_5$ in $\SPACE[O(S(n))]$. We now show that there exists $r\in\{1\cdot n^{-5}/2,2\cdot n^{-5}/2,\ldots,2n^2\cdot n^{-5}/2\}$ such that
    \begin{equation}\label{eq}
        \left|\E[\Bv]-k\cdot n^{-3}+r\right|> n^{-5}/6 \tag{$\star$}
    \end{equation}
    for every $k$ and $v$. There are $n^2$ different values $\E[\Bv]$ over $v$ in the vertex set $V(B)$ of the branching program, and for each $v$, there is at most one assignment to $r$ such that \eqref{eq} fails to hold for some $k\in\Z$. As there are $2n^2$ possible values for $r$, there must be one such that \eqref{eq} holds for all $k$ and $v$.
    
    Finally, let $j=O(S(|B|))$ be such that $\langle i\rangle(1^n,\Bv,r)$ halts using at most $j$ space for every $v$. Such a $j$ exists per assumption and the fact that the input $(1^n,\Bv,r)$ satisfies the promise of \Cref{prop:BPLcomplete} for every $v$. Thus, upon reaching the tuple $(i,j,r)$, the set of estimates $\tpv=\langle i\rangle(1^n,\Bv,r)$ must
    satisfy $\left|\tpv-\E[\Bv]\right|\leq n^{-3}$ for every $v\in V(B)$. Then running $\algLCTest(1^n,B,\{\tpv\}_{v\in V(B)})$ (where we wait for the test to request a particular value $\tpv$ and then recompute it from $\langle i\rangle$, avoiding the need to store all $n^2$ values) will result in \algLCTest~ accepting, and hence $\algUniv$ halts in the claimed space bound. Moreover, the returned value $\delta=\langle i\rangle(1^n,B,r)$ satisfies that $\left|\delta-\E[B]\right|\leq n^{-1}$.
\end{proof}
We finally prove the converse.
\begin{lemma}\label{lem:onlyif}
    For every space-constructible function $S:\N\ra \N$ satisfying $S(n)\geq \log n$, $\prBPL\subseteq \SPACE[O(S(n))]$ if $\algUniv$ runs in space $S(n)$.
\end{lemma}
\begin{proof}
    By \Cref{prop:BPLcomplete} it suffices to solve $f_3$ using a logspace reduction to $\algUniv$. Given $(1^n,B,r)$ as the input (where $r$ is the rounding threshold, which we will ignore), let $\delta:=\algUniv(1^n,B)$ be the value returned by $\algUniv$ on $B$. By \Cref{lem:decides} we have $|\delta-\E[B]|<n^{-1}$, and hence $\delta$ is a desired deterministic output for $f_3$.
\end{proof}

We can then conclude the proof of \Cref{thm:universal}.
\begin{proof}[Proof of~\Cref{thm:universal}]
    Let $\algUniv$ be the algorithm as defined above. \Cref{itm:univGood} follows from \Cref{lem:decides} (and the fact that it returns a value follows from \Cref{lem:halts}). The if direction of \Cref{itm:univIFF} follows from \Cref{lem:halts}, and the only if direction follows from \Cref{lem:onlyif}.
\end{proof}

Finally, we conclude the proof of \Cref{universal-derandomizer}.
\begin{proof}[Proof of~\Cref{universal-derandomizer}]
    Let $U$ be the algorithm that, given the description of a randomized logspace algorithm $R$ and an input $x$ where $|x|=n$, constructs (in deterministic logspace) the ordered branching program $B:=R(x,\cdot)$ of length and width at most $m=\poly(n)$ that represents the action of $R$ over its random bits. Then let $U$ call $\algUniv(1^m,B)$, and if the value returned is less than $1/2$ return $0$, and otherwise return $1$. By the promise on $R$ we have either $\Pr[B(U_n)=1]>3/4$ or $\Pr[B(U_n)=1]<1/4$, and as in both cases we estimate the expectation of $B$ up to error $1/n$ by \Cref{thm:universal}, we correctly decide which case we are in, and the space consumption follows from that of \Cref{thm:universal}.
\end{proof}

\section{Black Box Testing} \label{sec:BBT}
We now prove that hitting sets imply black-box two-sided derandomization of ordered branching programs. To do so, we first formally define hitting sets and deterministic samplers:
\begin{definition}
    Given a class of functions $\cF=\{f:\zo^n\ra\zo\}$, an \emph{$\eps$-hitting set generator (HSG)} $H:\zo^s\ra\zo^n$ for $\cF$ satisfies that for every $f\in \cF$ with $\E[f]\geq \eps$, there exists $y\in \zo^s$ where $f(H(y))=1$. We say $H$ is \emph{explicit} if there is a uniform algorithm that computes $H(x)$ in space $O(s)$ given $1^n$ and $x$.
\end{definition}

\begin{definition}
    Given a class of functions $\cF=\{f:\zo^n\ra\zo\}$, an \emph{$\eps$-(deterministic) sampler} $\SAMP$ with space complexity $s(n)$ for $\cF$ is a deterministic algorithm that runs in space $s(n)$ and, given oracle access to $f\in \cF$, makes queries to $f$ and outputs an estimate $\delta$ satisfying $|\delta-\E[f]|\leq \eps$.
\end{definition}
A deterministic sampler captures the idea of a derandomization algorithm that only accesses the branching program in a black-box fashion, and such a notion has been explored before in the context of small-space derandomization~\cite{HU22,CH20,PV22}.

We now give a formal statement of \Cref{thm:equiv:int}.
We state it in terms of dependence on the seed length of the HSG, as our result generically converts a hitting set to a sampler with comparable space complexity.
\newcommand{\cH}{\mathcal{H}}
\begin{theorem}\label{thm:equiv_formal}
    Suppose there is a uniformly constructible family $\cH=\{H_1,\ldots,\}$ where $H_n:\zo^{s(n)}\ra\zo^n$ is an explicit $1/2$-hitting set with seed length $s(n)$ for width $n$, length $n$ OBPs. Then there is a uniformly computable deterministic $\eps$-sampler with space complexity $O(s(nw/\eps))$ for width $w$, length $n$ OBPs.
\end{theorem}
We prove this by developing a local consistency test that can be implemented given black-box access to a branching program.
\subsection{Proof Overview}
The proof of \Cref{thm:equiv_formal} relies on developing a local consistency test that can be implemented given black-box access to a branching program (whereas all previous tests required access to the internal states of the program). We first describe how we can access the internal states of the program in a black-box manner. 
		
Given a branching program $B:\zo^n\ra\zo$ and a hitting set $H:\zo^s\ra\zo^n$, for each seed $x\in \zo^s$ and layer $i\in [n]$, the program reaches some state $v$ on input $H(x)_{1..i}$. We can index this state in a black-box fashion by writing down $(x,i)$. However, as potentially many seeds may reach the same state $v$, we would like to collapse these duplicates back together. Since we cannot examine layer $i$ of the program, we can instead attempt to test if $x$ and $x'$ reach the same state, by plugging in every HSG output and see if the programs starting from $(x,i)$ and $(x',i)$ behave differently.
\begin{definition}[Informal statement of~\Cref{def:indis}]
    For $x,x'\in \zo^s$ and $i\in [n]$, tuples $(x,i)$ and $(x',i)$ are \textit{indistinguishable} if for every $y\in\zo^s$,
    \[B(H(x)_{1..i}H(y)_{1..n-i})=B(H(x')_{1..i}H(y)_{1..n-i}).\]
\end{definition} 
It is not the case that indistinguishable tuples always reach the same state. However, Cheng and Hoza were able to show the following:
\begin{lemma}[\cite{CH20} (Informal)]\label{lem:soundIntro}
    Suppose states $v$ and $v'$ are reached by indistinguishable tuples. Then the probability of accepting in $B$ starting from $v$ is similar to that of accepting starting from $v'$.
\end{lemma}
Thus, the states have similar behavior from layer $i$ onward. Unfortunately, it is not the case that indistinguishable states always have indistinguishable out-edges. Thus, a naive attempt to learn the program using query access would print \textit{both} out-edges and output a nondeterministic branching program, a model that is provably $\NL$-hard to derandomize. It is likewise unclear how to select a single edge to print in a way that maintains the acceptance probability of the program. In the constant-width regime, Cheng and Hoza~\cite{CH20} circumvented this by remembering $O(\log n)$ bits of information about every state in layer $i+1$ while constructing layer $i$, allowing them to choose a good out-edge. However, this does not seem feasible for super-constant width.
Instead, we develop a local consistency test that can tolerate conflating indistinguishable states.

Suppose for every tuple $(x,i)$ we are given an estimate $\tp_{x,i}$, which is supposedly close to the true probability of accepting from $v:=B[\vs,H(x)_{1..i}]$. 
\begin{definition}[Black-Box Local Consistency Test (Informal)]\label{test:introbb}
    Given black-box access to an ordered branching program $B:\zo^n\ra\zo$ and a hitting set $H:\zo^s\ra\zo^n$ and estimates $\{\tp_{x,i}\}_{x\in \zo^s,i\in [n]}$, verify that the following conditions hold:
    \begin{enumerate}
        \item For every pair of indistinguishable tuples $(x,i),(x',i)$, we have $|\tp_{x,i}-\tp_{x',i}|\leq O(\eps)$.
        \item For every tuple $(x,i)$, let $(x_0,i+1)$ and $(x_1,i+1)$ be arbitrary tuples that are indistinguishable from $B[\vs,H(x)_{1..i}0]$ and $B[\vs,H(x)_{1..i}1]$ respectively. Then
        \[
        \left|\tp_{x,i}-\frac{\tp_{x_0,i+1}+\tp_{x_1,i+1}}{2}\right|\leq O(\eps).
        \]
    \end{enumerate}
    If all such conditions hold, output the estimate $\tp_{0,0}$, and otherwise reject.
\end{definition}
We think of all our tests as having a completeness and soundness component, where completeness means that a set of estimates which are sufficiently close to the true probabilities are gauranteed to pass, and soundness means that the test passing implies the returned estimate is close to the true value (where the precise parameters are discussed later).

It is not too difficult to the tests of \Cref{test:introbb} in space $O(s+\log n)$ given $H$ and black-box access to $B$, as we can enumerate over the seeds of the hitting set and layers in the program, and all such tests are \say{local}, in the sense that they deal with at most two layers and a constant number of seeds.

If every state is distinguishable from every other, and $H$ hits every state in the program, the test of \Cref{test:introbb} is equivalent to the following white-box local consistency test:
\begin{definition}[White-Box Local Consistency Test (Informal)]\label{test:introwb}~
    \begin{enumerate}
        \item For every $v$, all estimates of the accepting probability from $v$ must be within $\eps$ of each other.
        \item For every $v$, estimates of the accepting probability from $v,v_0:=B[v,0]$, and $v_1:=B[v,1]$ (which we denote $\widetilde{p_v},\widetilde{p_{v_0}}$, and $\widetilde{p_{v_1}}$) must satisfy $\widetilde{p_{v}}\approx (\widetilde{p_{v_1}}+\widetilde{p_{v_0}})/2$.
    \end{enumerate}
    If all such conditions hold, output an arbitrary estimate $\widetilde{p_{\vs}}$, and otherwise reject.
\end{definition}
The test of \Cref{test:introwb} clearly accepts if the estimates are exactly (or within $\eps/2$ of) the true probabilities of accepting from each vertex. Likewise, soundness is not difficult to show.
However, this idealized version of the test in \Cref{test:introbb} not exactly happen, for two reasons:
\begin{enumerate}
    \item We may impose Item (1) checks between tuples that reach different, yet indistinguishable, states, and likewise for the $0$ and $1$ states of Item (2).
    \item We may fail to impose Item (2) checks between $v$ and $B[v,0]$ and $B[v,1]$ if no string output by the hitting set reaches one of the latter states.
\end{enumerate}
Issue (1) must be dealt with in the proof of completeness (as we add some tests not in the white-box tester) and (2) in the proof of soundness (as we sometimes fail to impose tests that should be present). Issue (1) is the easier of the two to deal with. Since indistinguishable states have similar probability of accepting by \Cref{lem:soundIntro}, good estimates for the accepting probabilities of indistinguishable vertices will still be within $O(\eps)$ of each other. Issue (2), in contrast, seemingly presents a real issue for the soundness. For state $v:=B[\vs,H(x)_{1..i}]$ where there is no seed $x'$ where $B[\vs,H(x')_{1..i+1}]=B[v,0]$, we could run \textit{no} local consistency test to verify $\tp_{x,i}$. In fact, the estimate of the probability of accepting from $v$ could be arbitrarily wrong, and we would have no ability to detect it. However, we observe that every such $v$ has low probability of being reached from the start state. This is because if no ($\eps$-)HSG output reaches $B[v,0]$, $v$ must have probability of being reached from the start state at most $2\eps$. But then a very bad estimate of the probability of accepting from $v$ only changes the overall probability of accepting by at most $O(\eps)$ (and such an argument can be run for all non-verified states simultaneously). Ultimately, we are able to show that the lack of these checks can only increase the overall error by $O(\eps)$, which is tolerable.

Putting it all together, we show a black-box tester that, given estimates $\tp_{x,i}$ for the probability of accepting from $B[\vs,H(x)_{1..i}]$ for every $x$ and $i$, either outputs an approximation of the expectation of the program or rejects the input. To conclude, we use an idea of Cheng and Hoza to find a good set of estimates $\tp_{x,i}$ using a hitting set. First, to obtain a better result for nontrivial yet suboptimal hitting set generators, we slightly modify the tester to take in $n\cdot w$ estimates, corresponding (essentially) to an estimate for the acceptance probability from every state in the original branching program. Then we show (essentially using the argument of~\cite{CH20}), that there is a branching program $T$ of length $\poly(nw/\eps)$ and width $\poly(nw/\eps)$ that divides its input into $n\times w$ blocks, and uses the block labeled with $v$ as a long random string to estimate $\tp_{v}$ for every state $v$ in the program, and accepts if all these estimates are within $\eps$ of the true acceptance probability. The program uses the true probabilities to check if the empirical average of the samples is within $\eps$ of the true values, but we do not need to explicitly construct it - we only need that it exists, and hence our HSG family will contain some string hitting it. Finally, we argue that we can compute the associated empirical averages with oracle access to $B$, rather than $T$. A string that hits $T$ will produce good estimates $\tp_{v}$ for every $v$, and our black-box tester will accept on these estimates. Then we can simply enumerate over hitting set strings, and return the first accepted estimate.

\subsection{Black-Box Local Consistency Tests}
\newcommand{\pvr}{p_{v\ra}}
We now formally state the black-box local consistency test:
\begin{theorem}\label{thm:bb_lc}
    There is a deterministic space $O(s+\log n)$ algorithm that, given an explicit $\eps$-HSG $H:\zo^s\ra\zo^n$ for length $n$, width $w^2$ branching programs, oracle access to an OBP $B$ of width $w$ and length $n$, and estimates $\{\tp_{x,i}\}_{x\in \zo^s,i\in [n]}$, either outputs a value or rejects. Moreover:
    \begin{enumerate}
        \item If for every $x$, we have $|\pvr-\tp_{x,i}|\leq 2\eps$ where $v=B[\vs,H(x)_{1..i}]$ for every $i<n$ and $\tp_{x,n}=B(H(x))$, then the algorithm outputs a value.
        \item If the algorithm outputs $\delta$, then $|\E[B]-\delta|\leq 6\eps n$.
    \end{enumerate}
\end{theorem}
We remark that, despite this result being a black-box test versus the white-box local consistency test of Cheng and Hoza, it obtains an improved soundness loss (of $\eps n$ rather than $\eps nw$), which is relevant in the regime where the branching program has width much larger than length. This is notable as obtaining optimal error samplers in the Nisan-Zuckerman regime~\cite{NZ96} (where optimal-error hitting sets are already known) is a well known open question. Unfortunately, we do not obtain this result, as the argument that we can obtain good accepting probability estimates using a hitting set (\Cref{lem:goodest}) requires a hitting set for ordered programs of length $nw\gg n$.

We first define notation related to using $H$ to traverse the branching program:
\begin{definition}\label{def:indis}
    For every $x\in \zo^s$ and $i\in [n]$, let
    \[
    v_i(x) := B[\vs,H(x)_{1..i}].
    \]
    Note that this implies $v_i(x)=v_i(x')$ if $B[\vs,H(x)_{1..i}]=B[\vs,H(x')_{1..i}]$, i.e. the two seeds reach the same vertex in layer $i$. For convenience, we write $p_{x,i}:=p_{v_i(x)\ra}$. Moreover, for states $u,u'\in V_i$ we write $u\sim u'$ if the two states are \emph{indistinguishable} under $H$, i.e. for all $y$,
    \[B[u,H(y)_{1..n-i}]=B[u',H(y)_{1..n-i}].\] 
\end{definition}

We can now define the consistency test implemented by the algorithm.
\begin{definition}[Local Consistency Test]\label{def:test}
    Given $B$ and $H$ and the estimates $\tp_{x,i}$, let the test be as follows:
    \begin{enumerate}
        \item \label{t1}
        For every $x,i\in \zo^s\times [n]$ and for every $x_0,x_1\in \zo^s$ such that $B[v_i(x),0]\sim v_{i+1}(x_0)$ and $B[v_i(x),1]\sim v_{i+1}(x_1)$, require
        \[\left|\tp_{x,i} - \frac{\tp_{x_1,i+1}+\tp_{x_0,i+1}}{2}\right|\leq 5\eps.\]
        \item\label{t2}
        For every $x,x'\in \zo^s$ and $i\in [n]$ such that $v_i(x)\sim v_i(x')$, require $|\tp_{x,i}-\tp_{x',i}|\leq 5\eps.$
        \item\label{t3} For every $x\in \zo^s$, require $\tp_{x,n}=B(H(x))$.
    \end{enumerate}
\end{definition}
Note that given $H$ and oracle access to $B$ and the estimates $\tp_{x,i}$, we can compute all such tests in space $O(s+\log n)$. We first show this test is complete:
\begin{lemma}\label{lem:complete}
    Suppose for every $x$, $|\tp_{x,i}-p_{x,i}|\leq 2\eps$ for $i<n$ and $\tp_{x,n}=B(H(x))$. Then the test of~\Cref{def:test} passes.
\end{lemma}
To show this we require the following, which follows from arguments about the mass of the set difference. 
\begin{claim}[Lemma 3.2~\cite{CH20}]\label{clm:consis}
    If $v\sim v'$, then $|p_{v\ra}-p_{v'\ra}|\leq \eps$. 
\end{claim}
We can then prove completeness.
\begin{proof}[Proof of \Cref{lem:complete}]
    Consider an arbitrary \Cref{t1} test:
    \[
    \left|\tp_{x,i} - \frac{\tp_{x_1,i+1}+\tp_{x_0,i+1}}{2}\right|
    \]
    Let $v:=v_i(x)$ and for $b\in \zo$ let $v_b :=B[v,b]$ be the state actually reached following edge $b$ from state $v$. Then for $b\in \zo$ let $u_b:=v_{i+1}(x_b)$ be the state reached by $x_b$ in layer $i+1$. Note that $u_b$ does not necessarily equal $v_b$, as we could be conflating different states in layer $i+1$, but $u_b\sim v_b$. Thus:
    \begin{align*}
        \left|\tp_{x,i} - \frac{\tp_{x_1,i+1}+\tp_{x_0,i+1}}{2}\right| &\leq 4\eps+\left|p_{v\ra}- \frac{p_{u_0\ra}+p_{u_1\ra}}{2}\right| && \text{(Assumption)}\\
        &\leq 5\eps + \left|p_{v\ra}- \frac{p_{v_0\ra}+p_{v_1\ra}}{2}\right| && \text{(\Cref{clm:consis})}\\
        &= 5\eps.
    \end{align*}
    The proof of \Cref{t2} is analogous, again using \Cref{clm:consis}, and \Cref{t3} is immediate.
\end{proof}

We now show soundness. The key issue is dealing with states $v$ such that \textit{no} $x$ satisfies $v=v_i(x)$, because we cannot guarantee consistency for these states. However, these states are precisely those that the HSG fails to hit, which must mean they have low probability of being reached from the start vertex, and hence their estimates being wrong does only a small amount of harm.
\begin{lemma}\label{lem:sound}
    Suppose the test of~\Cref{def:test} passes with estimates $\tp_{x,i}$. Then $|\tp_{0,0}-p_{\vs\ra}|\leq 6\eps n$.
\end{lemma}
To prove \Cref{lem:sound}, we first define states that are not verified:
\begin{definition}
    For every $x\in \zo^s$ and $i<n$, let $v=v_i(x)$ be an \emph{unverified} state if there is some $b\in \zo$ such that there is no $x'\in \zo^s$ satisfying $B[v,b]=v_{i+1}(x')$, and otherwise let $v$ be \emph{verified}. Let $v_n(x)$ be verified for every $x$. Note that for an unverified state there still could be $x'$ such that $B[v,b]\sim v_{i+1}(x')$, but we do not use this in the proof of soundness.
\end{definition}
We first show that the probability of reaching an unverified state is small.
\begin{lemma}\label{lem:reachsmall}
    Let $T$ be the event of reaching an unverified state in $B$. Then $\Pr[T(U_n)=1]<2\eps$.
\end{lemma}
\begin{proof}
    Let $R$ be the width $w+1$ program that is the same as $B$ except it accepts if and only if we reach a state not hit by the HSG. We have $\Pr[R(U_n)=1]<\eps$ by the goodness of the HSG. Furthermore, conditioned on reaching an unverified state in $B$, we have probability at least $1/2$ of reaching a state the HSG does not hit. Thus, $\eps > \Pr[R(U_n)=1] \geq \Pr[T(U_n)=1]/2$.
\end{proof}

We now construct a branching program such that the estimates for unverified states are consistent with the true probabilities of these states.
\begin{lemma}\label{lem:qdef}
    There exists an ordered branching program $Q:\zo^n\ra\zo$ on a superset of the vertices of $B$ such that:
    \begin{enumerate}
        \item $|\E[Q]-\E[B]|\leq 2\eps $.
        \item $Q$ is identical to $B$ when restricted to edges between verified states, and edges from verified states to unverified states.
        \item For every unverified state $v$ in $B$, for every $x\in \zo^s$ such that $v=v_i(x)$, we have $|\tp_{x,i}-q_{v\ra}|\leq 5\eps$, where $q_{v\ra}$ is the probability of accepting from $v$ in $Q$.
    \end{enumerate}
\end{lemma}
\begin{proof}
    We first construct $Q$. Let $N$ be the set of unverified states of $B$. For every $v\in N$ in layer $i$, note that for every $x,x'\in \zo^s$ satisfying $v=v_i(x)=v_i(x')$ we have $|\tp_{x,i}-\tp_{x',i}|\leq 5\eps$ by \Cref{t2}. Let $q_{v\ra}$ be a number satisfying $|q_{v\ra}-\tp_{x,i}|\leq 5\eps$ for every such $x$. We now modify $B$ by wiring both edges from $v$ to a new (arbitrarily complex) set of states such that $v$ now has probability of accepting exactly $q_{v\ra}$,\footnote{Technically this may not be possible without making $Q$ a probabilistic branching program. However, this construction purely exists to analyze the probabilities $q_{v\ra}$, so we ignore this minor complication.} and we do this for every unverified $v$. Let $Q$ be this new branching program. It is clear by construction that $Q$ satisfies Property 2. Furthermore, by \Cref{lem:reachsmall} we have Property 1.
\end{proof}
We now prove \Cref{lem:sound} by showing that the estimates $\tp_{x,i}$ are consistent with the modified program $Q$.
\newcommand{\qvr}{q_{v\ra}}
\begin{lemma}\label{lem:qgood}
    Let $Q$ be defined as in~\Cref{lem:qdef}. Then for every $v$ in $B$ and every $x$ such that $v=v_i(x)$, we have
    $|\tp_{x,i}-\qvr|\leq 5\eps \cdot (n-i).$ In particular, $|\tp_{0,0}-q_{\vs\ra}|\leq 5\eps n$.
\end{lemma}
\begin{proof}
    The case $i=n$ holds by \Cref{t3} of \Cref{def:test} (and the fact that all final layer states are unmodified).
    Now assume this holds for layer $i+1$. Then for every $v=v_i(x)$ in layer $i$, we have two possibilities:
    \begin{itemize}
        \item Case 1: $v$ is unverified. In this case, $|\tp_{x,i}-\qvr|\leq 5\eps$ by~\Cref{lem:qdef}.
        \item Case 2: $v$ is verified. For $b\in \zo$, let $v_b:=B[v,b]$ and let $x_b$ be such that $v_b=v_{i+1}(x_b)$ (and note that such $x_0,x_1$ exist because we are not in Case 1). Then:
        \begin{align*}
            |\tp_{x,i}-\qvr| &= \left|\tp_{x,i}-\frac{q_{v_0\ra}+q_{v_1\ra}}{2}\right|\\
            &\leq 5\eps + \left|\frac{\tp_{x_0,i+1}+\tp_{x_1,i+1}}{2}-\frac{q_{v_0\ra}+q_{v_1\ra}}{2}\right| && \text{(\Cref{t1})}\\
            &\leq 5\eps + \frac{1}{2}\left|q_{v_0\ra}-\tp_{x_0,i+1}\right|+\frac{1}{2}\left|q_{v_1\ra}-\tp_{x_1,i+1}\right|\\
            &\leq 5\eps\cdot (n-i) && \text{(Induction.)} \qedhere
        \end{align*} 
    \end{itemize}
\end{proof}
We now conclude the proof of \Cref{lem:sound}.
\begin{proof}[Proof of \Cref{lem:sound}]
    We have:
    \begin{align*}
        |\E[B]-\tp_{0,0}| &\leq 2\eps + |\E[Q]-\tp_{0,0}| && \text{(\Cref{lem:qdef})}\\
        &\leq 2\eps + 5\eps n && \text{(\Cref{lem:qgood}).} \qedhere
    \end{align*}
\end{proof}

We now conclude \Cref{thm:bb_lc}.
\begin{proof}[Proof of~\Cref{thm:bb_lc}]
    Given $H:\zo^s\ra\zo^n$ and the estimates $\{\tp_{x,i}\}_{x\in \zo^s,i\in [n]}$, we run the tests as specified in \Cref{def:test}. All such tests can be implemented in space $O(s+\log n)$, as we now explain. Given $x\in \zo^s$ and $b\in \zo$, we can determine if $x'\in \zo^s$ satisfies $B[\vs,H(x)_{1..i}b]\sim v_{i+1}(x')$ by enumerating over $y\in \zo^s$ and computing 
    \[\bigwedge_{y\in \zo^s}\mathbb{I}[B(H(x)_{1..i}(x)bH(y)_{1..n-i-1}) = B(H(x')_{1..i+1}H(y)_{1..n-i-1})].\] 
    This can be implemented in space $O(s+\log n)$ given black-box access to $B$, and hence we can determine which \Cref{t1} tests to run in the desired space bound, and similar reasoning applies to the \Cref{t2} and \Cref{t3} tests.

    Finally, if all such tests pass, output $\tp_{0,0}$. By \Cref{lem:complete} we have that the completeness condition holds, and by \Cref{lem:sound} we have that the soundness condition holds.
\end{proof}

\subsection{Putting It All Together}
We now prove \Cref{thm:equiv_formal} from \Cref{thm:bb_lc}. It remains to show that we can generate a good set of estimates $\{\tp_{x,i}\}$ using a hitting set. We first show that we can modify \Cref{thm:bb_lc} to only take in $nw$ estimates, rather than $n\cdot 2^s$. This is not required for \Cref{thm:equiv:int}, but it improves the parameters in the case that $H$ is a highly nontrivial yet non-optimal hitting set.

To do so, we first observe that the indistinguishably relation induces a set of equivalence classes on the seeds:
\begin{definition}\label{def:equiv_class}
    Given an OBP $B$ of width $w$ and length $n$ and $H:\zo^s\ra\zo^n$, for every $i$ let $C_{1,i},\ldots,C_{w,i}\subset \zo^s$ be the equivalence classes (with possibly some empty) of $\zo^s$ under the indistinguishably relations $x\sim_i x'$ iff $v_i(x)\sim v_i(x')$. Let $y_{j,i}\in\zo^s$ be the lexicographically first element of $C_{j,i}$, and WLOG assume $y_{1,i}< y_{2,i}< \ldots< y_{w,i}$ for every $i$. Moreover, let $v_{j,i}:=v_i(y_{j,i})$. Note that given $B$ and $H$, $y_{j,i}$ (and hence an HSG output that reaches $v_{j,i}$) is constructible in space $O(s+\log n)$ given $i,j$.
\end{definition}

\begin{corollary}\label{cor:smalltester}
    There is a deterministic space $O(s+\log n)$ algorithm that, given an explicit $\eps$-HSG $H:\zo^s\ra\zo^n$ for length $n$, width $w^2$ branching programs, oracle access to an OBP $B$ of width $w$ and length $n$, and estimates $\{\tp_{j,i}\}_{j\in [w],i\in [n-1]}$, either outputs a value or rejects. Moreover:
    \begin{enumerate}
        \item If $|p_{v_{j,i}\ra}-\tp_{j,i}|\leq \eps$ for every $j,i$, where $v_{j,i}$ is as defined in \Cref{def:equiv_class}, then the algorithm outputs a value.
        \item If the algorithm outputs $\delta$, $|\E[B]-\delta|\leq 6\eps n$.
    \end{enumerate}
\end{corollary}
\begin{proof}
    The tester simply takes in the estimate $\tp_{j,i}$ for $p_{v_{j,i}\ra}$, copies it to be the estimate for every seed in the $j$th equivalence class for layer $i$, perfectly computes $p_{x,n}:=B(H(x))$ for every $x\in \zo^s$, and runs \Cref{thm:bb_lc}. Clearly if the tester outputs a value it is within $6\eps n$ of $\delta$, as we simply restrict the inputs to \Cref{thm:bb_lc}. Furthermore, note that for an arbitrary $v:=v_i(x)$ in equivalence class $C_{j,i}$, we have by \Cref{clm:consis}:
    \[
    |p_{v\ra}-\tp_{j,i}|\leq \eps+|p_{v_{j,i}\ra}-\tp_{j,i}|\leq 2\eps
    \]
    and hence if $|p_{v_{j,i}\ra}-\tp_{j,i}|\leq \eps$ is satisfied for all $i$ and $j$ we satisfy the completeness condition of \Cref{thm:bb_lc}, and so the tester will return a value.
\end{proof}

We now argue that there is a hitting set string that can be used to produce good estimates for $p_{v_{j,i}\ra}$, where $v_{j,i}$ is as defined in \Cref{def:equiv_class}.
The argument that such an output exists is a straightforward modification of the proof in Cheng and Hoza~\cite{CH20} that there exists an HSG output inducing estimates that satisfy their local consistency test.
\newcommand{\EST}{\mathtt{EST}}
\begin{lemma}[\cite{CH20}]\label{lem:goodest}
    For every OBP $B$ of length $n$ and width $w$ and $H:\zo^s\ra\zo^n$ and $\eps>0$, there exists $t=O(\log(nw)/\eps^2)$ and an OBP $\EST:\zo^{n\times w\times tn}\ra \zo$ of length and width $\poly(nw/\eps)$ defined as follows:
    \[
    \EST(z_{1,1},\ldots,z_{w,n}) = \bigwedge_{i\in [n],j\in [w]}\EST_{j,i}(z_{j,i})
    \]
    where $\EST_{j,i}(z_{j,i})$ computes as follows. It interprets $z_{j,i}$ as $t$ samples of length $n$, and computes
    \[
    \EST_{j,i}(s_1,\ldots,s_t) = \mathbb{I}\left[\Pr_{k\in [t]}[B[v_{j,i},s_k]=\va] \in [p_{v_{j,i}\ra}-\eps,p_{v_{j,i}\ra}+\eps]\right].
    \]
    where $v_{j,i}$ is as defined in \Cref{def:equiv_class} in terms of $H$).
    Then $\E[\EST]>1/2$, and for every $z$ such that $\EST(z)=1$, for every $j,i$, using the samples in block $z_{j,i}$ of $z$ to estimate the acceptance probability from $v_{j,i}$ produces an estimate with at most $\eps$ additive error.
\end{lemma}
\begin{proof}[Proof Sketch]
    It is clear that $\EST_{j,i}$ can be computed by an ordered branching program of the claimed length and width, by duplicating the subprogram of $B$ starting from $v_{j,i}$ and counting the number of satisfied trials using an additional $O(\log(t))$ bits of memory, and accepting if the final count is within the specified range. Thus the conjunction $\EST$ can be computed in the claimed space bound. We then choose $t$ sufficiently large such that a random input satisfies all these checks with overwhelming probability. We note that $\EST$ is defined in terms of the exact probabilities of acceptance, which the tester does not have, but we only need that the program exists, not that we can construct it.
\end{proof}

Then the proof of \Cref{thm:equiv_formal} follows.
\begin{proof}[Proof of \Cref{thm:equiv_formal}]
    By a standard reduction (see e.g.~\cite{CH20}), $\cH$ implies an explicit family of $\eps$-hitting sets for length $n$, width $w$ OBPs with seed length $s(\poly(nw/\eps))=O(s(nw/\eps))$ (where the final equality follows as for any $s(n)=\Omega(\log^2 n)$ the theorem is trivial by the fact that the Nisan PRG exists, so we may assume this is not the case).
    
    Let $H$ be a $\eps/6n$-HSG for length $n$ and width $w^2$ OBPs with seed length $O(s(nw/\eps))$. Let $H_2$ be a $1/3$-HSG for length $n^2wt=\poly(nw/\eps)$, width $(nw/\eps)^c$ OBPs, where $t$ is as in \Cref{lem:goodest} with $\eps:=\eps/6n$. 
    By choice of parameters, $H_2$ has seed length $s_2:=O(s(nw/\eps))$. The sampler enumerates over every $z\in \zo^{s_2}$. For every such $z$, the sampler calls the tester of \Cref{cor:smalltester}, and when an estimate $\tp_{j,i}$ for $p_{v_{j,i}\ra}=\E_{r\sim_R\zo^{n-i}}[B[v_{j,i},r]]$ is required by the tester, we use the $j,i$ block of $H_2(z)$ (as done by $\EST_{j,i}$ in \Cref{lem:goodest}) to compute the estimate. Note that we can find $y_{j,i}$, the lexicographically first seed in equivalence class $j$ in layer $i$, in logspace, and by definition $v_{j,i}=v_i(y_{j,i})$. Thus, we can compute $\tp_{j,i}$ by enumerating over the samples $s_1,\ldots,s_t$ in block $j,i$ in $H_2(z)$ and returning 
    \[\E_{k\in [t]}[B(H(y_{j,i})_{1..i}s_k)] = \E_{k\in [t]}[B[v_{j,i},s_k]].
    \]
    If the tester accepts, return the value that the tester outputs, and otherwise increment $z$. The space complexity is $O(\log(n)+s_2)$ by composition of space-bounded algorithms.

    Now suppose the sampler returns a value. By Item 2 of \Cref{thm:bb_lc}, the returned estimate is within $(\eps/6n)6n=\eps$ of the true expectation. 
    To show the sampler returns a value, note that by Item 1 of \Cref{thm:bb_lc} it suffices to argue that we give the tester a series of inputs $\{\tp_{j,i}\}$ such that $|p_{v_{j,i}\ra}-\tp_{j,i}|\leq (\eps/6n)$ for every $i,j$. But these are precisely the estimates generated by a string $x$ such that $\EST(x)=1$, and $H_2$ hits this program by choice of parameters, so we conclude.
\end{proof}

\bibliographystyle{alpha}
\bibliography{ref}

\end{document}